\pdfoutput=1
\documentclass[11pt,a4paper]{article}
\usepackage[T1]{fontenc}
\usepackage[utf8]{inputenc}
\usepackage[english]{babel}
\usepackage{amsmath}
\usepackage{amsfonts}
\usepackage{amssymb}
\usepackage{graphicx}
\usepackage{lmodern}
\usepackage[left=2cm,right=2cm,top=2cm,bottom=2cm]{geometry}
\usepackage{authblk}
\usepackage{titlesec}

\usepackage[group-separator={,}, group-minimum-digits=4]{siunitx}
\usepackage{csquotes}
\usepackage{natbib}
\let\cite\citet
\let\citet\citet
\let\citep\citep

\title{Adaptive MCMC for multiple changepoint analysis with applications to large datasets}
\author{Alan Benson\textsuperscript{1,2,$\star$} and Nial Friel\textsuperscript{1,2}}

\makeatletter
\renewcommand\@date{{%
  \vspace{-\baselineskip}%
  \large\centering
  \smallskip
  \textsuperscript{1}School of Mathematics and Statistics, University College Dublin\par
  \vspace{0.5ex}
  \textsuperscript{2}Insight Centre for Data Analytics\par
  \vspace{0.5ex}
  \textsuperscript{$\star$}alan.benson@insight-centre.org

  \bigskip
  \bigskip

  \today
}}
\makeatother

\usepackage[ruled,noend,linesnumbered,algosection]{algorithm2e}

\usepackage{bm}
\newcommand{\thetavec}{\bm{\theta}}

\newcommand{\zvec}{\bm{z}}
\newcommand{\deltavec}{\bm{\delta}}
\newcommand{\deltavecpr}{\bm{\delta}^{\prime}}
\newcommand{\zvecpr}{\bm{z}^{\prime}}

\newcommand{\datavec}{\bm{y}}

\newcommand{\norm}[1]{\left\lVert#1\right\rVert}
\makeatletter
\g@addto@macro\normalsize{%
  \setlength\abovedisplayskip{13pt}
  \setlength\belowdisplayskip{13pt}
  \setlength\abovedisplayshortskip{13pt}
  \setlength\belowdisplayshortskip{13pt}
}
\makeatother

\usepackage{tikz}

\usepackage{caption}
\usepackage[pdfencoding = auto,colorlinks,citecolor=blue]{hyperref}

\usepackage{amsthm}
\newtheorem{theorem}{Theorem}[section]
\newtheorem*{remark}{Remarks}
 
\theoremstyle{definition}
\newtheorem{definition}{Definition}[section]

\usepackage{cleveref}

\begin{document}
% Title
\maketitle

% Abstract
\begin{abstract}
We consider the problem of Bayesian inference for changepoints where the number and position
of the changepoints are both unknown. In particular, we consider product partition models
where it is possible to integrate out model parameters for the regime between each
changepoint, leaving a posterior distribution over a latent vector indicating the
presence or not of a changepoint at each observation. The same problem setting has been considered by 
Fearnhead (2006) where one can use filtering recursions to make exact inference. 
However, the complexity of this filtering recursions algorithm is quadratic in the number of observations. 
Our approach relies on an adaptive Markov Chain Monte Carlo (MCMC) method for finite
discrete state spaces. We develop an adaptive algorithm which can learn from the past states of 
the Markov chain in order to build proposal distributions which can quickly discover where 
changepoint are likely to be located. We prove that our algorithm leaves the posterior 
distribution ergodic. Crucially, we demonstrate that our adaptive MCMC algorithm is viable for 
large datasets for which the filtering recursions approach is not. Moreover, we show that
inference is possible in a reasonable time thus making Bayesian changepoint detection computationally efficient.
\end{abstract}

% Introduction
\section{Introduction}
Changepoint problems arise in many practical instances in statistics, for example, signal processing, financial economics, process monitoring control and DNA sequence analysis. 
Here we consider chronologically ordered data over a period of time where it is suspected that there may have been some change(s) in the underlying generating process. 
For changepoints in parametric models, a parameter value (e.g.\ Gaussian mean or Gaussian precision) applicable to a certain time period may not extend well to another time period.
Some examples include the rate of occurrences of coal mining disasters during the 18th and 19th century \citep{raftery1986}, gene expression sequences \citep{hocking2014seganndb} and 
financial time series \citep{chen2011parametric}. In this paper it is shown that analysis of multiple changepoint problems is feasible for larger datasets in a Bayesian setting using adaptive MCMC.

Markov Chain Monte Carlo methods (MCMC) can be used to estimate changepoint locations conditional on a \emph{fixed} number of changepoints, \citet{stephens1994} presents an MCMC method for this problem.
When the number of changepoints is unknown, inference is more challenging. This is the problem which we address in this paper. 
A common approach for state-space dimension traversing is the reversible jump algorithm of \citet{green1995reversible} which performs trans-dimensional MCMC over a set of models, each incorporating a different number of changepoints. A drawback of this algorithm is that it can be difficult to design proposals so that the chain mixes well within and well between all available models. An alternative approach due to \citet{chib1998estimation} compares models with different numbers of changepoints using approximate Bayes Factors from the MCMC output in a post-processing step. The latter method requires MCMC model output for each number of changepoints under consideration.

\citet{raey} developed a clever forward-backward algorithm, filtering recursions, which allows one to sample exactly from the posterior distribution of changepoints. The filtering recursions share some similarity to 
product partition models \citep{barry1992}. The overwhelming advantage of this method is that once the filtering recursions have been calculated, it allows one to draw to be sampled from the posterior using 
Carpenter's algorithm \citep{carpenter1999improved} that exploits the exponentially distributed spacing of order statistics in a uniform distribution. 

However, a drawback of filtering recursions is that the algorithm requires a precomputation step to compute the recursions which has a time complexity that is quadratic in the number of observations and thus 
restricts the amount of data that can be used to perform efficient inference in a reasonable time. \citet{raey} offers an solution to this problem that lowers the precision of the recursions in order to make 
their calculation time approximately linear in the number of observations and the price to pay is it results in an approximate algorithm thereby. 

Adaptive Markov Chain Monte Carlo Methods (AMCMC) have recently emerged in an attempt to improve the efficiency of MCMC algorithms.
Typically adaptive MCMC uses \emph{on-the-fly} refinement of the proposal distribution, taking information from the past history of the MCMC chain to yield a better mixing algorithm.
The adaptive Metropolis algorithm of \citet{haario2001} was one of the earliest adaptive MCMC algorithms using a random walk Metropolis algorithm with an adapted covariance matrix. It is limited to continuous state spaces and to target distributions where a Gaussian proposal is suitable.

Adaptive MCMC methods on discrete state spaces have not yet been widely studied, yet these are very well suited to this methodology. This is because the design of adaptable proposals on discrete state spaces has the advantage that discrete state spaces carry the property of smallness, outlined in \citet{meyn2012markov}, so that simultaneous uniform ergodicity of the proposal kernels is guaranteed, provided that the state space is irreducible and the transition kernel is aperiodic. 
The second condition necessary for ergodicity of adaptive MCMC, diminishing adaptation, can be satisfied in many ways on a discrete state space leading to widely applicable methods in problems such as variable selection and Bayesian optimisation \citep{mahendran2012adaptive}.
\citet{griffin2014individual} presents an adaptive MCMC algorithm on a discrete state space to carry out variable selection in a model choice setting.

In recent years, the emergence of big data across a vast range of models in statistics and machine learning has lead to the need for methods that can  scale well to large datasets. We highlight how our adaptive changepoint approach scales well with an increasing number of observations and an increasing number of changepoints. The size of large datasets can present challenges for non-adaptive MCMC due to the presence of many local modes in the posterior distribution. We show empirically how are algorithm learns to move away from local modes which hinder MCMC.

The remainder of the paper is organised as follows. Section \ref{sec:mcmodels} describes multiple changepoint models in a Bayesian framework, Section \ref{sec:model} describes our adaptive changepoint sampler and introduces some advanced adaptation techniques which improve efficiency. We present a brief review in Section \ref{sec:recur} of filtering recursions \citep{raey}. Section \ref{sec:proof} provide a proof of our algorithm and results for three datasets are presented in Section \ref{sec:results} along with comparisons to filtering recursions.

All methods in this paper have been implemented in \texttt{C} using the Intel \texttt{C} compiler running on an Intel i7 3.40GhZ equipped machine with 16GB of RAM. Code is available on request from the authors.

\section{Multiple Changepoint models} \label{section:multiplechangepointmodel}
\label{sec:mcmodels}
Consider observed data $\bm{y} = (y_1, y_2, \dots, y_n)$, where observation $y_i$ is observed before observation $y_j$, for $i<j$. We model $\bm{y}$ such that each observation $y_i$ arises independently from a likelihood model depending on a parameter $\theta_i \in \Theta$ whose value may or may not change from one observation to the next. The points at which $\theta_i$ does change are called changepoints.

Consider the possibility of an unknown $k < n$ changepoints in $\bm{y}$ occurring at positions $\bm{\tau} = \{\tau_1, \tau_2, \dots, \tau_k \}$.
These changepoints partition $\bm{y}$ into $k+1$ contiguous non-overlapping segments 
\begin{equation}
\left\{
(y_1, y_{\tau_1}), (y_{\tau_1 + 1}, y_{\tau_2}), \dots, (y_{\tau_k + 1}, y_{n})
\right\}.
\end{equation}

This partitioning of $\bm{y}$ can be represented by a fixed length latent changepoint indicator vector $\zvec = \{z_1, z_2, \dots,z_{n-1}\}$ with $z_t^{} = 1$ for each $t \in \bm{\tau}$ and $z_t^{} = 0$ for each $t \notin \bm{\tau}$, with the number of changepoints satisfying $ k = \sum_{i=1}^{n-1} z_i$.
Within segment $j$, the likelihood has a constant parameter $\theta_j,\text{ } 1 \leq j \leq k+1$. The full likelihood across all segments can be expressed as a product of $k+1$ segment likelihoods
\begin{equation}
f(\bm{y} \vert \theta_1, \theta_2, \dots, \theta_{k+1}, \bm{z}) = \prod\limits_{j=1}^{k+1} \prod_{i=\tau_{j-1}+1}^{\tau_j} f(y_i \vert \theta_j)
\label{likelihood}
\end{equation}
where $\tau_0 = 0, \tau_{k+1} = n$ and where $f(y_i\vert\theta_j^{})$ denotes the likelihood of observation $y_i$ in a segment with parameter $\theta_j$.
In a Bayesian formulation the joint posterior distribution for the latent changepoint indicator vector $\zvec$ and segment parameters $\thetavec = \{\theta_1,\dots,\theta_{k+1}\}$ can be written as a product of the full segment likelihood \eqref{likelihood} and the priors for $\zvec$ and $\thetavec$,
\begin{equation}
\begin{split}
\pi(\zvec, \thetavec \vert \datavec)
&\propto
f(\datavec \vert \thetavec, \zvec) 
\pi(\thetavec \vert \zvec)
\pi(\zvec)
\\[1.5ex]
&= \left(\prod\limits_{j=1}^{k+1} \prod\limits_{i = \tau_{j-1} + 1}^{\tau_j} f(y_i \vert \theta_j)\right)
\left( \prod\limits_{j=1}^{k+1} \pi(\theta_j) \right)
\pi(\zvec). \label{eq:jointztheta}
\end{split}
\end{equation}
The dependence of $\thetavec$ on $\zvec$ is only through the prior multiplicity of $\thetavec$ ($k+1$) which sets the dimension of the prior term $\pi(\thetavec \vert \zvec)$.
This shares some similarity to the hierarchical changepoint model used in \citet{green1995reversible} except that it does not condition on the number of changepoints and so this varies over the support of $\zvec$.

The prior for $\bm{z}$ specifies how the changepoint positions should be distributed prior to the data being observed.
A convenient form that captures the gap lengths between changepoints is,
\begin{equation*}
\pi(\zvec) = \pi(\tau_1, \dots \tau_k) = 
g_0(\tau_1)\left( \prod_{j=2}^{k}{g(\tau_j - \tau_{j-1})}\right)\left(1 - G(n - \tau_k) \right),
\end{equation*}
where $g_0(\cdot)$ is the distribution of the distance to the first changepoint, $g(\cdot)$ is the gap distribution for the distance between successive changepoints and $G(\cdot)$ is the cumulative distribution function for $g(\cdot)$. The choice for $g$ can be a negative binomial or its special case, a geometric distribution
\begin{equation*}
g(t) = \binom{t-1}{k-1}p^k(1-p)^{t-k}, \hspace{2em} g_0(t) = p(1-p)^{t-1}.
\end{equation*}
A more complex prior that minimises the \textit{a priori} clustering of changepoints \citep{green1995reversible}, is specified by the distribution of even order statistics of a draw of size $2k+1$ from $(1, \dots, n-1)$ without replacement. This prior prevents changepoints occurring at adjacent observations which minimises outliers (degenerate changepoints) being classified as true changepoints.

The priors for $\theta_j$ can be chosen to be conjugate to the likelihood, however this is not a requirement. The next section details the collapsing of the joint posterior \eqref{eq:jointztheta} when the prior is conjugate but if it is possible to collapse the joint posterior using another method (e.g.\ quadrature) this is also feasible for use in our algorithm.

% End Multiple Changepoint models

\subsection{Collapsing multiple changepoints models}
\label{sec:collapselike}
We assume that it is possible to integrate (collapse) out $\thetavec = \{ \theta_1, \theta_2, \dots, \theta_{k+1}\}$ parameters from the posterior \eqref{eq:jointztheta} to leave a discrete state space of changepoint positions. This is also the approach taken by \citet{raey}. With an appropriate conjugate prior for $\thetavec$, the resulting posterior for $\zvec$ is
\begin{equation}
\begin{split}
\pi(\zvec \vert \datavec) &\propto 
\int_{\thetavec}
f(\datavec \vert \thetavec, \zvec) 
\pi(\thetavec)
\pi(\zvec)
\,d \thetavec
\\[1.5ex]
&= \pi(\zvec) \prod\limits_{j=1}^{k+1}
\left(
\int_{\theta_j}  \prod\nolimits_{i = \tau_{j-1} + 1}^{\tau_j} f(y_i \vert \theta_j) \pi(\theta_j)
\,d \theta_j
\right) \\[1.5ex]
&=
\pi(\zvec) \prod\limits_{j=1}^{k+1} \mathrm{P}(\tau_{j-1} + 1, \tau_{j}),
\label{marglike}
\end{split}
\end{equation}
where $\mathrm{P}(\tau_{j-1} + 1, \tau_{j}) = \int_{\theta_j}  \prod\nolimits_{i = \tau_{j-1} + 1}^{\tau_j} f(y_i \vert \theta_j) \pi(\theta_j)
\,d \theta_j$ denotes the evidence for segment $(y_{\tau_{j-1}+1}, y_{\tau_j})$. The evidence (marginal likelihood) is the probability of the data observed in that segment after the dependence on the parameter $\theta_j$ has been integrated out with respect to its prior. The dependence of $\thetavec$ on $\zvec$ has been removed the position of changepoints and the within segment parameter are assumed independent.%using the parameter priors.
\subsubsection{A simple example of Collapsing - Poisson Gamma}
Consider the case where the data in segment $j$ can be modelled by a Poisson distribution with parameter $\theta_j > 0$. Placing a Gamma($\alpha, \beta$) prior on each $\theta_j$ and integrating out $\theta_j$ for $j \in \{ 1, \dots, k+1 \}$, the marginal likelihood for segment $(y_a, y_b)$ is, 
\begin{equation}
\begin{split}
\mathrm{P}(a, b) &= \int_{0}^{\infty}
\frac{\beta^\alpha}{\Gamma(\alpha)} \theta_j^{\alpha - 1} e^{-\alpha \theta_j} \prod\limits_{i = a}^{b} \frac{\theta_j^{y_i}}{y_i!} e^{-\theta_j} \, d \theta_j \label{poismarg} \\[1.5ex]
&= \frac{\beta^\alpha}{\Gamma(\alpha)} \frac{1}{F_{a:b}} \frac{\Gamma(S_{a:b} + \alpha)}{\left(b - a + 1 + \beta\right)^{S_{a:b} + \alpha}},
\end{split}
\end{equation}
where
\begin{equation*}
\mathrm{F}_{a:b} = \prod\nolimits_{i=a}^{b} y_i! \hspace{2em} \text{and} \hspace{2em}
\mathrm{S}_{a:b} = \sum\nolimits_{i = a}^{b} y_i.
\end{equation*}
Precomputation of $\mathrm{F}_{1:t}^{}$ and $\mathrm{S}_{1:t}^{}$ for $1 \leq t \leq n$ and using the following recursions
\begin{equation*}
\mathrm{F}_{a:b} = \frac{\mathrm{F}_{1:b}}{\mathrm{F}_{1:{a-1}}} \hspace{2em} \text{and} \hspace{2em} \mathrm{S}_{a:b} = \mathrm{S}_{1:b} - \mathrm{S}_{1:{a-1}},
\end{equation*}
negates the need to store each individual $y_i$ for computation of the marginal likelihood \eqref{poismarg} which is required to be computed many times in filtering recursions and in our algorithm.
Similar precomputations are available for other likelihood models, see Appendix \ref{app:normallike} for the Gaussian distribution mean and precision.

\section{Adaptive MCMC changepoint sampler}
\label{sec:model}
We now introduce our adaptive MCMC changepoint sampler to sample from the posterior distribution of changepoints \eqref{marglike}. \citet{wyse2010simulation} developed an MCMC scheme based on adding, deleting 
and position adjusting changepoints using samplers similar to those used by \citet{lavielle2001application}. The algorithm of \citet{wyse2010simulation} turns out to be a special case of our adaptive algorithm 
when no adaptation occurs and as we shall see, the adaptive MCMC algorithm we develop offers an improvement in efficiency, by comparison.

Sampling over $\bm{z}$ is a challenging problem as the size of the space scales exponentially with $n$, leaving brute force enumeration of all $\bm{z}$ intractable. However, for datasets with few changepoints 
($k \ll n$) the realised $\bm{z}$ vectors will be quite sparse. The design of our algorithm motivates searching element-wise through $\bm{z}$ identifying which elements (positions) are likely changepoints and 
those which are not. Positions which are deemed unlikely to be changepoints will tend not to be proposed as changepoint locations and conversely locations which are identified as being locations of changepoints
will tended to be proposed more frequently. 
%have their proposal probability decreased for future proposal moves and elements that are likely to be changepoints will have their 
%proposal probability increased for future proposal moves. 
In this way, our algorithm will facilitate proposed moves to centre around areas of high changepoint activity and move away from areas of low changepoint activity. As we will shortly see, this adaptive algorithm
where proposed changepoint locations change over time will by design preserve the ergodicity of the adaptive Markov chain.
%The increases 
%and decreases to the proposal probabilities will preserver ergodicity by design of our algorithm. 

We now describe the adaptive algorithm in detail and defer a proof of ergodicity to Section \ref{sec:proof}.

\subsection{Detailed description}
\label{sec:amcmcdesc}
At iteration $t$ denote the current state of changepoint locations as $\bm{z}^{(t)}$.
Our algorithm consists of three proposal moves to update the vector $\bm{z}^{(t)}$. The three proposal moves involve adding a new changepoint to $\bm{z}^{(t)}$ (\emph{add move}), deleting a changepoint from $\bm{z}^{(t)}$ (\emph{delete move}) and moving an existing changepoint within $\bm{z}^{(t)}$ (\emph{adjust move}). At each iteration $t$, one of either the \emph{add move} or the \emph{delete move} is selected with probability $p$ and $1-p$, respectively. The \emph{adjust move} is performed after either an add or delete move and in our implementation of the adaptive MCMC algorithm is an optional move type. %It could also be performed at every iteration or alternatively with a certain probability $p_\text{adjust}$ at each iteration.

The space of all realisable $\bm{z}$ vectors is large, having $2^{n-1}$ elements. It is important therefore to add changepoints in locations of high posterior changepoint probability and delete changepoints 
in areas of low posterior changepoint probability. It turns out that adaptively learning these areas \emph{on-the-fly} provides a route to a scalable inferential framework for large datasets, as we now illustrate.

We associate with $\bm{z}^{(t)}$ two iteration dependent selection weight vectors $\bm{a}^{(t)} = \{a^{(t)}_1 \dots, a^{(t)}_{n-1}\}$ and $\bm{d}^{(t)} = \{d^{(t)}_1 \dots, d^{(t)}_{n-1}\}$. 
We remark that these weights are correspond to how often the algorithm should pick a particular point. This is different to the approach of \citet{griffin2014individual} where the vectors are used as inclusion probabilities for variable selection. If a changepoint is proposed to be added, a position $i$ (having $z^{(t)}_i = 0$) will be selected as the add position with probability $a^{(t)}_i / \sum_{\{j, z_j = 0\}} a^{(t)}_j$. If a changepoint is proposed to be deleted, some position $i$ (having $z^{(t)}_i=1$) will be selected as the deletion position with probability $d^{(t)}_i / \sum_{\{j, z_j = 1\}} d^{(t)}_j$. If the relevant add or delete move is accepted then the selected element $i$ of $\bm{z}^{(t+1)}$ will be toggled, otherwise $\bm{z}^{(t+1)}$ does not change from $\bm{z}^{(t)}$.

The probability of accepting or rejecting the moves described above will depend on the relative change in the marginal likelihood of the segment added or deleted around position $i$. Let $a$ be the changepoint immediately before $i$ and $b$ the changepoint immediately after $i$. The addition of a changepoint at position $i$ would cause the segment that contains position $i$ to be split into two new segments $(y_{a+1},y_{i})$ and $(y_{i+1},y_{b})$. The deletion of a changepoint at position $i$ would cause the two segments created by the changepoint at $i$ to merge into one segment $(y_{a+1}, y_b)$. All other segments remain the same. The marginal likelihood ratios are thus
\begin{equation}
\text{Add Move} \rightarrow \frac{\mathrm{P}(a+1, i)\mathrm{P}(i+1, b)}{\mathrm{P}(a+1,b)} \quad\quad \text{Delete Move} \rightarrow \frac{\mathrm{P}(a+1,b)}{\mathrm{P}(a+1, i)\mathrm{P}(i+1, b)}.
\end{equation}
The two moves are summarised clearly in Figure \ref{fig:moves}. The \emph{adjust move} simply selects uniformly some $z_i = 1$ and propose to move it locally somewhere between the changepoint before it and 
the changepoint after it. If there are no changepoints this move cannot be and is not attempted. %Details of this move are shown in \citet{wyse2010simulation}.

\begin{figure}[ht]
\centering
\begin{minipage}[t]{0.485\textwidth}
  \vspace{0pt}  
  \renewcommand{\algorithmcfname}{Move}
  \begin{algorithm}[H]
	\SetAlgoHangIndent{0pt}
    \caption{Add a Changepoint}
    Calculate $\bm{a}_+^{(t)} = \sum\nolimits_{\{j, z_j = 0\}} \bm{a}_j^{(t)}$ and $\bm{d}_+^{(t)} = \sum\nolimits_{\{j, z_j = 1\}} \bm{d}_j^{(t)}$.\\
    Select $i$ with $z_i = 0$ with prob. $\bm{a}_i^{(t)}/\bm{a}_+^{(t)}$.\\
    Accept to toggle $z_i = 1-z_i$ with probability $\min(1, \alpha_{\text{add}})$, where
    $$\alpha_{\text{add}} =
    \tfrac{\pi(\bm{z}^{\prime})}{\pi(\bm{z}^{})}
	\tfrac{\mathrm{P}(a+1,i)\mathrm{P}(i+1,b)}{\mathrm{P}(a+1,b)}  
	\tfrac{1-p}{p} 
	\tfrac{\bm{d}_i^{(t)}/(\bm{d}_i^{(t)} + \bm{d}_+^{(t)})}{\bm{a}_i^{(t)}/\bm{a}_+^{(t)}}.
    $$
  \end{algorithm}
\end{minipage}\hspace{1em}
\begin{minipage}[t]{0.485\textwidth}
  \vspace{0pt}
  \renewcommand{\algorithmcfname}{Move}
  \begin{algorithm}[H]
  	\SetAlgoHangIndent{0pt}
    \caption{Delete a Changepoint}
    Calculate $\bm{d}_+^{(t)} = \sum\nolimits_{\{j, z_j = 1\}} \bm{d}_j$ and $\bm{a}_+^{(t)} = \sum\nolimits_{\{j, z_j = 0\}} \bm{a}_j^{(t)}$.\\
    Select $i$ with $z_i = 1$ with prob. $\bm{d}_i^{(t)}/\bm{d}_+^{(t)}$.\\
    Accept to toggle $z_i = 1-z_i$ with probability $\min(1, \alpha_{\text{del}})$, where
    $$\alpha_{\text{del}} = 
    \tfrac{\pi(\bm{z}^{\prime})}{\pi(\bm{z}^{})}
	\tfrac{\mathrm{P}(a+1,b)}  {\mathrm{P}(a+1,i)\mathrm{P}(i+1,b)}
	\tfrac{p}{1-p} 
	\tfrac{\bm{a}_i^{(t)}/(\bm{a}_i^{(t)} + \bm{a}_+^{(t)})}{\bm{d}_i^{(t)}/\bm{d}_+^{(t)}}.
    $$
  \end{algorithm}
\end{minipage}
\caption{Adaptive MCMC changepoint sampler moves, the add move is performed with probability $p$ and the delete move with probability $1-p$.}
\label{fig:moves}
\end{figure}

This is the basis of our changepoint sampler. We are now left to describe the adaptation scheme used to update the $\bm{a}^{(t)}$ and $\bm{d}^{(t)}$ vectors during the algorithm, using the past history of the add and selected moves. This is a crucial part of the algorithm as these parameters decide where to place changepoints and remove changepoints in an efficient manner. This is described in the following section.

\subsection{Adaptation of the selection weights \texorpdfstring{$\bm{a}^{(t)}$}{a} and \texorpdfstring{$\bm{d}^{(t)}$}{d}}
The MCMC algorithm of \citet{wyse2010simulation}
 selects positions $i$ for addition and deletion uniformly at random from all the valid $n-1$ positions. This is equivalent to having constant vectors $\bm{a}^{(t)}$ and $\bm{d}^{(t)}$ which do not vary with iteration $t$. The adaptive method we use, proposes to update $\bm{a}^{(t)}$ and $\bm{d}^{(t)}$ using information from previously accepted \emph{add} and \emph{delete} moves. The scheme for adaptation is given in Figure \ref{fig:adaptscheme}. The strategy is to target the acceptance rate of the \emph{add} and \emph{delete} moves to an overall target acceptance rate $\alpha_\text{target}$ by updating the $\bm{a}^{(t)}$ and $\bm{d}^{(t)}$ at each iteration. The updates are performed on the $\log$ scale to ensure that the weights remain positive. % To prevent zero weights occurring a constant offset parameter $\epsilon$ is employed for each update. This offset is also required for the proof in section \ref{sec:proof}.%
\begin{figure}[ht]
\fbox{\parbox{0.98\textwidth}{
\vspace{1ex}
\textbf{Adaptation Scheme}\\
At iteration $t$:
\begin{enumerate}
\item If an \emph{add} move at point $i$ has been accepted then update only the $a^{(t)}_i$ parameter as follows
\begin{equation*}
\log(a_{i}^{(t+1)}) = \log(a_{i}^{(t)}) + \frac{h}{t/n}\left( \alpha_{\text{add}}^{} - \alpha_\text{target} \right).
\end{equation*}
\item If a \emph{delete} move at point $i$ has been accepted then update only the $d^{(t)}_i$ parameter as follows
\begin{equation*}
\log(d_{i}^{(t+1)}) = \log(d_{i}^{(t)}) + \frac{h}{t/n}\left( \alpha_{\text{del}}^{} - \alpha_\text{target} \right).
\end{equation*}
\end{enumerate}
\hspace{0.5em}\textbf{Parameters}\par
\hspace{0.5em}$h$ - Initial Adaptation ($h > 0$)\par
\hspace{0.5em}$t/n$ - Monte Carlo time, iterations ($t$) per number of datapoints ($n$)\par\par
}}
\caption{Adaptation scheme to update the vectors $\bm{a}^{(t)}$ and $\bm{d}^{(t)}$. }
\label{fig:adaptscheme}
\end{figure}

This adaptation scheme is different from \citet{griffin2014individual} in that there is no restriction on $0 < \bm{a}_i < 1$ or $0 < \bm{d}_i < 1$ as these are unnormalised selection weights and not probabilities. The parameter $h$ controls the initial intensity of the adaptation, we find values $ << 1$ work well. 

%The parameter $\lambda$ controls the overall speed of the adaptation, smaller values cause the adaptation to perform slower, however $\lambda \in (0.5, 1\rbrack$ since otherwise it doesn't satisfy the square intregrability ($\sum \left(\frac{h}{t^\lambda}\right)^2 < \infty$)  conditions of stochastic approximation \citep{robbins1951stochastic}

\subsubsection{A note on Non Uniform Sampling for selection weights}
The $\bm{a}^{(t)}$ and $\bm{d}^{(t)}$ weights, once normalised appropriately using $\bm{a}_+^{(t)}$ and $\bm{d}_+^{(t)}$ (see Figure \ref{fig:moves}), must be sampled from to propose elements of $\bm{z}^{(t)}$ for toggling. Discrete random variate generation for non-uniform probability vectors presents an extra level of complexity. In the case of  \citet{wyse2010simulation} with no adaptation, selection of elements for toggling is $\mathcal{O}(1)$ and is extremely efficient. To take advantage of the adaptive proposals the algorithm requires an efficient non-uniform sampler.

A na\"{\i}ve implementation of non-uniform sampling from the $\bm{a}^{(t)}$ and $\bm{d}^{(t)}$ vectors involves building a cumulative distribution of the values, $\mathcal{O}(n)$ time, and then sampling from this by binary lookup, $\mathcal{O}(\log_2 n) \text{ time}$. This is significantly slower and may even detriment the use of the adaptive algorithm in the first instance.
A method due to \citet{walker1974new} overcomes this problem by precomputing lookup tables called alias tables in $\mathcal{O}(n)$ time and then sampling in $\mathcal{O}(1)$ time. A numerically stable implementation of Walker's method that overcomes numerical errors is due to \citet{vose1999simple}. A discussion of the alias method is given in Appendix \ref{app:alias}.

Using alias tables we can get quite close to uniform sampling efficiency. Note that \citet{matias1993dynamic} allows updating alias tables in less than $\mathcal{O}(n)$ time, however this imposes restrictions on the magnitude of the change in weights at each adaptation step.

\subsection{Advanced Adaptation Techniques}
\label{sec:adapt}
In this section some advanced techniques are presented to improve the efficiency of the adaptive method. 
It is possible to implement thresholding of the $\bm{a}^{(t)}$ values so that only some of the values use alias tables. Dual adaptation is used by \citet{griffin2014individual} to simultaneously update $\bm{\bm{a}}^{(t)}$ and $\bm{\bm{d}}^{(t)}$ after an accepted move. We modify this to our adaptation scheme. These advanced techniques allow the algorithm to be computationally efficient while performing the adaptive updates. Many issues with adaptive MCMC can arise due to adapting too quickly. These issues are discussed in \citet{latuszynski2014containment}.
\subsubsection{Advanced Adaptation 1: Thresholding of non-changepoints}
Many of the $\bm{a}_i$ values won't significantly change in magnitude over the course of the algorithm. This is due to the update of the $\bm{a}_i$ values only being performed on acceptance of a changepoint and for points far away from changepoints the $\bm{a}_i$ will rarely change. Computational time is still spent embedding these small $\bm{a}_i$ in the rebuilding of alias tables each time any  $\bm{a}_i$ changes. This problem isn't as pronounced for the $\bm{d}_i$ values as we assume that there are many more non-changepoints than changepoints in a dataset.

To take advantage of the low number of changepoints, we propose to split the points that are not changepoints into two groups, one with high posterior probability of being added, $G_\text{active}$, and the other with a low posterior probability of being added, $G_\text{inactive}$. The membership of each group is mutually exclusive and is determined by a threshold parameter $\bm{a}_{\text{cutoff}}$. All points begin in $G_\text{inactive}$ and as the $\bm{a}_i$ values are adapted, points with $\bm{a}_i > \bm{a}_{\text{cutoff}}$ move to $G_\text{active}$. The other points remain in $G_\text{inactive}$ and are assumed to have a flat weight of $\bm{a}_{\text{inactive}} < \bm{a}_{\text{cutoff}}$ which means they can be sampled without the use of alias tables (equivalent to uniform sampling within $G_\text{inactive}$). Each element of $G_\text{inactive}$ will retain it's true underlying $\bm{a}_i$ value but this will only be used for sampling if and when it moves into $G_\text{active}$. The thresholding will modify the algorithm slightly and the modifications to the acceptance probabilities are show in Figure \ref{thres_alg}.
\begin{figure}[ht]
\begin{minipage}[t]{0.485\textwidth}
  \vspace{0pt}  
  \renewcommand{\algorithmcfname}{Move}
  \begin{algorithm}[H]
  	  	\SetAlgoHangIndent{0pt}
    \caption{Add (with threshold)}
    Calculate $\bm{a}^{(t)}_\text{active} = \sum\nolimits_{\{j\vert z_j = 0, j \in G_\text{active}\}} \bm{a}^{(t)}_j$ and $\bm{d}^{(t)}_+ = \sum\nolimits_{\{j, z_j = 1\}} \bm{d}^{(t)}_j$\\
    Select $i$ with $z_i = 0$ with prob. $\bm{a}^{(t)}_i/\bm{a}^{(t)}_+$\\
    Accept to toggle $z_i = 1-z_i$ with probability $\alpha_{\text{add}} = \min(1, \hat{\alpha}_{\text{add}})$
    $$
	\begin{aligned}    
    \hat{\alpha}_{\text{add}} = 
    &\tfrac{\pi(\bm{z}^{\prime})}{\pi(\bm{z}^{})}
	\tfrac{\mathrm{P}(a+1,i)\mathrm{P}(i+1,b)}{\mathrm{P}(a+1,b)}  
	\tfrac{1-p}{p} \\
	&\times\tfrac{\bm{d}^{(t)}_i/(\bm{d}^{(t)}_i + \bm{d}^{(t)}_+)}{\widehat{\bm{a}}^{(t)}_i/(\bm{a}^{(t)}_\text{active} + \bm{a}_{\text{inactive}}\lvert G_\text{inactive} \rvert)}
	\end{aligned}
    $$
    where $\widehat{\bm{a}}^{(t)}_i = \bm{a}^{(t)}_i$ if $i \in G_\text{active}$ or $\widehat{\bm{a}}^{(t)}_i = \bm{a}_{\text{inactive}}$ otherwise.
  \end{algorithm}
\end{minipage}\hspace{1em}
\begin{minipage}[t]{0.485\textwidth}
  \vspace{0pt}
  \renewcommand{\algorithmcfname}{Move}
  \begin{algorithm}[H]
  	  	\SetAlgoHangIndent{0pt}
    \caption{Delete (with threshold)}
    Calculate $\bm{d}^{(t)}_+ = \sum\nolimits_{\{j, z_j = 1\}} \bm{d}^{(t)}_j$ and $\bm{a}^{(t)}_\text{active} = \sum\nolimits_{\{j\vert z_j = 0, j \in G_\text{active}\}} \bm{a}^{(t)}_j$\\
    Select $i$ with $z_i = 1$ with prob. $\bm{d}^{(t)}_i/\bm{d}^{(t)}_+$\\
    Accept to toggle $z_i = 1-z_i$ with probability $\alpha_{\text{add}} = \min(1, \hat{\alpha}_{\text{add}})$
    $$
\begin{aligned}    
    \alpha_{\text{del}} = 
    &\tfrac{\pi(\bm{z}^{\prime})}{\pi(\bm{z}^{})}
	\tfrac{\mathrm{P}(a+1,b)}  {\mathrm{P}(a+1,i)\mathrm{P}(i+1,b)}
	\tfrac{p}{1-p} \\
	&\times\tfrac{\widehat{\bm{a}}^{(t)}_i/(\widehat{\bm{a}}^{(t)}_i + \bm{a}^{(t)}_\text{active} +  \bm{a}_{\text{inactive}}\lvert G_\text{inactive} \rvert)}{\bm{d}^{(t)}_i/(\bm{d}^{(t)}_+)}
\end{aligned}
    $$
    where $\widehat{\bm{a}}^{(t)}_i = \bm{a}^{(t)}_i$ if $i \in G_\text{active}$ or $\widehat{\bm{a}}^{(t)}_i = \bm{a}_{\text{inactive}}$ otherwise.
  \end{algorithm}
\end{minipage}
\caption{Adjusted moves for use with thresholding of $\bm{a}_i$ values. Note that $\lvert G_\text{inactive} \rvert$ denotes the cardinality of the inactive set.}
\label{thres_alg}
\end{figure}

\subsubsection{Advanced Adaptation 2: Dual adaptation}
As can be seen in the description of the moves, knowledge of $\alpha_{\text{add}}$ allows one to also calculate $\alpha_{\text{del}}$ quite easily. \citet{griffin2014individual} uses this idea to perform a double or dual adaptation of both $\bm{\bm{a}}^{(t)}$ and $\bm{\bm{d}}^{(t)}$ at each acceptance in the algorithm rather than updating only one of these vectors. The dual adaptation approach is applied without thresholding to the updates in Figure \ref{fig:adaptscheme} and is described in Appendix \ref{app:dual}.

\section{The alternative approach using filtering recursions}
\label{sec:recur}
\citet{raey} provides a filtering recursions approach to inferring changepoint positions. \citet{barry1992} have also used these type of recursive methods for analysis of changepoint problems. We give a brief recap of the filtering recursions method and we will use the method as a comparison to our adaptive changepoint sampler. Some drawbacks of the filtering recursions will also be discussed.

Define for $t=2,\dots,n$
\begin{equation*}
Q(t) = \mathbf{P}(y_t, \dots y_n \vert \text{changepoint at } t-1)
\end{equation*}
and $Q(1) = \mathbf{P}(y_1, \dots y_n)$. \citet{raey} provides a backward recursion for $Q(t)$ as follows, using the marginal likelihood $\mathrm{P}(a,b)$ in \eqref{marglike},
\begin{equation*}
Q(t) =\left(\sum_{i=t}^{n-1} g(i-t+1) \mathrm{P}(t,i) Q(i+1)\right) + \mathrm{P}(t,n) (1-G(n-t)).
\end{equation*}
The function $g(\cdot)$ is the gap length distribution between changepoints (for example, geometric) and $G(\cdot)$ is its cumulative distribution function.

Once the $Q(t)$ values have been calculated (normally on the log scale) it is possible to draw sample of size $N$ from the posterior distribution of positions as follows:
\begin{enumerate}
\item Initialise all $N$ samples to have a changepoint at $t=0$, i.e.\ $\tau_0 = 0$
\item For $t=0, \dots, n-2,$
\begin{enumerate}
\item Find $n_t$, the number of samples for which the last changepoint was at time $t$.
\item If $n_t > 0$, compute the probability distribution for the next changepoint
\begin{equation}
\mathbf{P}(\tau_j \vert \tau_{j-1}) = \mathrm{P}(\tau_{j-1} + 1, \tau_j)Q(\tau_j+1)g(\tau_j - \tau_{j-1})/Q(\tau_{j-1} + 1).
\label{transprob}
\end{equation}
\item Sample $n_t$ times, using Carpenter's algorithm (see Appendix \ref{app:carpenter} for details), from $\mathbf{P}(\tau_j \vert \tau_{j-1})$ and update the $n_t$ samples using a random permutation of the $n_t$ samples.
\end{enumerate}
\end{enumerate}
The filtering recursions approach has the advantage that the design of the method allows one to draw independently from the posterior distribution. Moreover Carpenter's algorithm for sampling the changepoints is fast. This method however has some drawbacks which arise as the dataset increases in size. Firstly, the calculation of the $Q(t)$ values is $\mathcal{O}(n^2)$ as the recursion for each possible ordered pair of points ($i < j$) must be computed before perfect simulation can begin. This calculation time can be reduced by truncating the $Q(t)$ sums once they fail to grow by a certain amount, \citet{raey} suggests $10 \times 10^{-10}$ and we compare various truncation levels in the results section.
The price to pay for this reduced run time is that the truncation introduces an approximation to the recursion algorithm. 
Secondly, hyperparameters must remain fixed throughout the algorithm as a change in hyperparameters or indeed the inclusion of a hyperprior would require complete recalculation of  $Q(t)$. Thirdly, for 
larger datasets ($\approx$ \SI{260000} observations for the largest example considered in this paper) the transition 
probabilities in \eqref{transprob} have the potential to become numerically unstable, as we outline in Section~\ref{sec:diffrecur}. We suggest using the exact algorithm, where possible. However for larger ($>$ \SI{100000} observations) datasets we advocate the use of our adaptive changepoint sampler as it is much more stable, by comparison.

\section{Proof of ergodicity for the Adaptive MCMC algorithm} \label{section_proof}
\label{sec:proof}
There are two parts to proving ergodicity for an adaptive MCMC algorithm on a discrete state space $\mathcal{X}$. The first establishes the notion of simultaneous uniform ergodicity and the second establishes diminishing adaptation. An adaptive MCMC algorithm which satisfies both of these conditions is ergodic by Theorem 1 of \citet{rosenthal2007coupling}.
\subsection{Simultaneous Uniform Ergodicity}
We first recap the definition of uniform ergodicity for a Markov chain, the equivalent Doeblin's condition and simultaneous uniform ergodicity for transition kernels on a state space $\mathcal{X}$.
\begin{definition}{(Uniform ergodicity)}
	\label{def:uniformergodicity}
A Markov chain on a state space $X$ with a transition kernel $P(x , \cdot)$ is called \textit{uniformly ergodic} if 
\begin{equation*}
\sup_{x \in \mathcal{X}}  \norm{P^{n}(x , \cdot) - \pi()}_{\text{TV}} \rightarrow 0 \text{ as } n \rightarrow \infty.
\end{equation*}
where $\norm{\cdot}_{\text{TV}}$ is the total variation norm.
\end{definition}
An equivalent definition by Theorem 16.0.2 \citep{meyn2012markov} states that there exists some $r>1$ and $R < \infty$ such that $\forall x \in \mathcal{X}$
	\begin{equation*}
\norm{P^{n}(x , \cdot) - \pi()}_{\text{TV}} \leq Rr^{-n}.
\end{equation*}
This implies that the convergence takes place at a geometric rate independent of the starting point $x_0 \in \mathcal{X}$ of the algorithm.

Uniform ergodicity is generally difficult to prove directly using Definition \ref{def:uniformergodicity}.  Instead uniform ergodicity can be often more easily checked by equivalence to Doeblin's condition on $\mathcal{X}$. This equivalence is shown in Theorem 16.0.2 of \citet{meyn2012markov} and is repeated here.

\begin{theorem}{(Doeblin's Condition)}
Suppose that Doeblin's Condition holds (as defined in \citet[p 396]{meyn2012markov}) so that there exists a probability measure $\phi$ on the measurable space $(\mathcal{X},\sigma\{\mathcal{X}\})$ with the property that for some $m$, some constant measure $\rho < 1$, some $\beta > 0$ and for a set $A \in \sigma\{\mathcal{X\}}$
\begin{equation*}
\phi(A) > \rho \Longrightarrow \mathrm{P}^{m}(x, A) > \beta
\end{equation*}
then the chain under transition kernel $\mathrm{P}^{m}(x, \cdot)$ is \textbf{uniformly ergodic}.
\end{theorem}

\begin{proof}
See Theorem 16.2.3 of \citet{meyn2012markov} and relevant lemmas.
\end{proof}
\noindent
Finally \citet{rosenthal2007coupling} define the notion of simultaneous uniform ergodicity for a collection of transition kernels indexed by $\gamma \in \Gamma$. This definition is repeated here.

\begin{definition}{(Simultaneous Uniform Ergodicity)}
A collection of transition kernels indexed by $\gamma \in \Gamma$ exhibit simultaneous uniform ergodicity if $\forall \, \gamma \in \Gamma $ and $\forall \, x \in X$
\begin{equation*}
\norm{P_{\gamma}^{n}(x , \cdot) - \pi()}_\text{TV} \leq R_{\gamma}r_{\gamma}^{-n}, \text{ where } R_{\gamma} < \infty \text{ and } r_{\gamma} > 1 \text{ for all } \gamma \in \Gamma
\end{equation*}
where $\norm{\cdot}_{\text{TV}}$ is the total variation norm.
\end{definition}
\begin{remark}
The uniform ergodicity parameters $R_\gamma$ and $r_\gamma$ for each kernel may depend on $\gamma$ but not on the states $x \in \mathcal{X}$ as otherwise uniform ergodicity would not hold.
\end{remark}
Verifying multiple Doeblin's Conditions is equivalent to verifying uniform ergodicity for all kernels $P_{\gamma}^{n}(x , \cdot)$. This in turn guarantees simultaneous uniform ergodicity.
\noindent
We will now prove simultaneous uniform ergodicity for the adaptive changepoint sampler.
\begin{theorem}{(Simultaneous uniform ergodicity of the adaptive changepoint sampler)} Let $\bm{\Gamma}^{(t)} = (\bm{a}^{(t)},\bm{d}^{(t)})$ be the set of adaptive weights at iteration $t$ and let $\bm{z}^{(t)}$ be the current state of the chain. Then for all $t$ the transition kernel using the weights $\bm{\Gamma}^{(t)}$,  $\mathrm{P}_{\bm{\Gamma}^{(t)}}(\bm{z}^{(t)}, \cdot)$ is uniformly ergodic.
\end{theorem}
\begin{proof}
As we are working over a discrete state space, $\bm{Z}$, the lower bound of the transition kernel $\mathrm{P}_{\bm{\Gamma}^{(t)}}(\bm{z}^{(t)}, \cdot)$ can be used as the value of $\beta$ in order to satisfy Doeblin's Condition. Denote the overall minimum value of any element of $\bm{a}^{(t)}$ or $\bm{d}^{(t)}$ by $\epsilon > 0$. The existence of this minimum follows from the adaptation scheme in Figure \ref{fig:adaptscheme} where it is not possible for any $\bm{a}^{(t)}$ or $\bm{d}^{(t)}$ to reach 0 when started from a positive value. Take the measure $\phi(\bm{z})$ in Doeblin's Condition to be the posterior distribution of $\bm{z}$, $\pi(\bm{z} \vert y)$. The value $\phi(\bm{z})$ is always positive as the prior for $\bm{z}$ allows for all $2^{n-1}$ values of $\bm{z}$ to occur with non-zero probability.

The 1-step transition moves of our algorithm are the add and delete moves, i.e.\ $m=1$ in Doeblin's Condition. An $m$-step kernel can be achieved by iteration of the $1$-step kernel $m$ times, Doeblin's condition only requires the existence of some $m$ and the aperiodicity of the chain.  Under add or delete moves the $1$-step transition kernel at iteration $t$, $\mathrm{P}_{\bm{\Gamma}^{(t)}}^{m=1}(\bm{z}^{(t)}, \bm{z}^{\prime})$, can be separated into its proposal and acceptance parts
\begin{equation}
\begin{split}
&\mathrm{P}_{\bm{\Gamma}^{(t)}}^{1}(\bm{z}^{(t)}, \bm{z}^{\prime}) \\[1.5ex] 
&= q_{\bm{\Gamma}^{(t)}}^{}(\bm{z}^{(t)}, \bm{z}^{\prime}) \alpha_{\bm{\Gamma}^{(t)}}^{}(\bm{z}^{(t)}, \bm{z}^{\prime}) 
+ \delta \left( \bm{z}^\prime - \bm{z}^{(t)}\right) \{ 1- \sum_{\bm{z}^\prime \neq \bm{z}^{(t)}} 
q_{\bm{\Gamma}^{(t)}}^{}(\bm{z}^{(t)}, \bm{z}^{\prime}) \alpha_{\bm{\Gamma}^{(t)}}^{}(\bm{z}^{(t)}, \bm{z}^{\prime} ) 
\}
\label{eq:fullkernel}
\end{split}
\end{equation}
where $\delta_{\bm{z}^{(t)}} \left( \bm{z}^\prime \right)$ is 1 if and only if $\bm{z}^{(t)}$ is identical to $\bm{z}^{\prime}$ in every element (i.e.\ the Hamming distance between current and proposal $\bm{z}$ is 0). By design of our algorithm, $\delta_{\bm{z}^{(t)}} \left( \bm{z}^\prime \right)$ = 0 as a different vector is always proposed by the add and delete moves and we only need consider the first term of \eqref{eq:fullkernel}
\begin{equation}
\mathrm{P}_{\bm{\Gamma}^{(t)}}^{1}(\bm{z}^{(t)}, \bm{z}^{\prime}) = q_{\bm{\Gamma}^{(t)}}^{}(\bm{z}^{(t)}, \bm{z}^{\prime}) \alpha_{\bm{\Gamma}^{(t)}}^{}(\bm{z}^{(t)}, \bm{z}^{\prime}) 
\end{equation}
The proposal kernel $q_{\bm{\Gamma}^{(t)}}^{}(\bm{z}^{(t)}, \bm{z}^{\prime}) \geq \dfrac{\epsilon}{\omega^{(t)}} >0$, where $\omega^{(t)}$ normalises the $\bm{a}^{(t)}$ or $\bm{d}^{(t)}$ weights depending on which of the add or delete is taking place, therefore
\begin{align*}
\mathrm{P}_{\bm{\Gamma}^{(t)}}^{1}(\bm{z}^{(t)}, \bm{z}^{\prime}) &\geq 
\dfrac{\epsilon}{\omega^{(t)}} \alpha_{\bm{\Gamma}^{(t)}}^{}(\bm{z}^{(t)}, \bm{z}^{\prime})
\\[1.5ex]
&=\dfrac{\epsilon}{\omega^{(t)}}
\min \left\{1, 
\frac{
\pi(\bm{z}^{\prime} \vert y) q_{\bm{\Gamma}^{(t)}}^{}(\bm{z}^{\prime}, \bm{z}^{(t)})
}
{
\pi(\bm{z}^{(t)} \vert y) q_{\bm{\Gamma}^{(t)}}^{}(\bm{z}^{(t)}, \bm{z}^{\prime})
} \right\}
\\[1.5ex]
&=\dfrac{\epsilon}{\omega^{(t)}}\pi(\bm{z}^{\prime} \vert y) q_{\bm{\Gamma}^{(t)}}^{}(\bm{z}^{\prime}, \bm{z}^{(t)})
\min \left\{\frac{1}{\pi(\bm{z}^{\prime} \vert y) q_{\bm{\Gamma}^{(t)}}^{}(\bm{z}^{\prime}, \bm{z}^{(t)})}, 
\frac{
	1
}
{
	\pi(\bm{z}^{(t)} \vert y) q_{\bm{\Gamma}^{(t)}}^{}(\bm{z}^{(t)}, \bm{z}^{\prime})
} \right\}
\\[1.5ex]
&\geq \dfrac{\epsilon}{\omega^{(t)}} \pi(\zvecpr \vert y) q_{\bm{\Gamma}^{(t)}}^{}(\zvecpr, \bm{z}^{(t)}).
\intertext{This inequality holds since $\pi(\bm{z}^{(t)} \vert y) q_{\bm{\Gamma}^{(t)}}^{}(\bm{z}^{(t)}, \bm{z}^{\prime}) < 1$ and $\pi(\bm{z}^{\prime} \vert y) q_{\bm{\Gamma}^{(t)}}^{}(\bm{z}^{\prime}, \bm{z}^{(t)}) < 1$.}
&\geq \dfrac{\epsilon}{\omega^{(t)}}\min_{\bm{z}} \pi(\bm{z} \vert y) \left(\frac{\epsilon}{\omega^{(t)}} \right) \\[1.5ex]
&= \left(\dfrac{\epsilon}{\omega^{(t)}}\right)^2 \min_{\bm{z}} \pi(\bm{z} \vert y) := \beta^{(t)} > 0.
\end{align*}

This verifies that Doeblin's condition holds for the $1$-step proposal kernel under any fixed set of adaptive weights $\bm{\Gamma}^{(t)}$ which is sufficient to prove simultaneous uniform ergodicity.
\end{proof}
\clearpage
\subsection{Diminishing Adaptation}
The second part of the proof is to verify diminishing adaptation for $\mathrm{P}_{\bm{\Gamma}^{(t)}}(\bm{z}, \cdot)$, $\forall t$. Recall the definition of diminishing adaptation \citep{rosenthal2007coupling}
\begin{definition}{(Diminishing adaptation)}
A series of transition kernels indexed by time $t$, $\mathrm{P}_{\bm{\Gamma}^{(t)}} (\bm{z}, \cdot)$, are said to obey diminishing adaptation if
\begin{equation*}
\lim_{t \rightarrow \infty} \sup_{\bm{z}} \norm{\mathrm{P}_{\bm{\Gamma}^{(t+1)}} (\bm{z}, \cdot) - \mathrm{P}_{\bm{\Gamma}^{(t)}} (\bm{z}, \cdot)} = 0.
\end{equation*}
\end{definition}
For this section of the proof, the two other definitions needed are the concept of Lipschitz and bi-Lipschitz continuity of a real-valued function.
\begin{definition}{(Lipschitz continuity)}
A function $f$ is Lipschitz if there exists $K > 0$ such that
\begin{equation*}
\lvert f(x_1) - f(x_2) \rvert \leq K \lvert x_1 - x_2 \rvert.
\end{equation*}
\end{definition}
\noindent
By the Mean Value Theorem this is equivalent to the function $f$ having a bounded first derivative.
\begin{definition}{(bi-Lipschitz continuity)}
A function $f$ is bi-Lipschitz if $f$ and its inverse $f^{-1}$ are both Lipschitz and thus one has
\begin{equation*}
\frac{1}{K} \lvert x_1 - x_2 \rvert \leq \lvert f(x_1) - f(x_2) \rvert \leq K \lvert x_1 - x_2 \rvert.
\end{equation*}
where $K > 0$ is the Lipschitz constant of $f$ and the inverse constant of $f^{-1}$.
\label{eq:bilip}
\end{definition}

\begin{theorem}{}
The adaptive changepoint sampler satisfies diminishing adaptation.
\end{theorem}
\begin{proof}
	For a 1-step move from $\bm{z} = \bm{z}^{(t)}$ to $\bm{z}^{\prime}$ define $\Delta^{(t)}$ to be the difference in the transition kernels between iteration $t$ and $t+1$. Again we only need consider the first term of \eqref{eq:fullkernel} in this difference.
\begin{equation}
\begin{split}
\Delta^{(t)} &= \mathrm{P}_{\bm{\Gamma}^{(t+1)}} (\bm{z}, \bm{z}^{\prime}) - \mathrm{P}_{\bm{\Gamma}^{(t)}} (\bm{z}, \bm{z}^{\prime}) \\[1.5ex]
&= q_{\bm{\Gamma}^{(t+1)}}^{}(\bm{z}, \bm{z}^{\prime}) \alpha_{\bm{\Gamma}^{(t+1)}}^{}(\bm{z}, \bm{z}^{\prime}) - q_{\bm{\Gamma}^{(t)}}^{}(\bm{z}, \bm{z}^{\prime}) \alpha_{\bm{\Gamma}^{(t)}}^{}(\bm{z}, \bm{z}^{\prime}) \\[1.5ex]
&= q_{\bm{\Gamma}^{(t+1)}}^{}(\bm{z}, \bm{z}^{\prime})
\min
\left\{1, 
\frac{
\pi(\bm{z}^{\prime} \vert y) q_{\bm{\Gamma}^{(t+1)}}^{}(\bm{z}^{\prime}, \bm{z})
}
{
\pi(\bm{z} \vert y) q_{\bm{\Gamma}^{(t+1)}}^{}(\bm{z}, \bm{z}^{\prime})
} \right\}
- q_{\bm{\Gamma}^{(t)}}^{}(\bm{z}, \bm{z}^{\prime})
\min
\left\{1, 
\frac{
\pi(\bm{z}^{\prime} \vert y)q_{\bm{\Gamma}^{(t)}}^{}(\bm{z}^{\prime}, \bm{z})
}
{
\pi(\bm{z} \vert y) q_{\bm{\Gamma}^{(t)}}^{}(\bm{z}, \bm{z}^{\prime})
} \right\} \\[1.5ex]
&=
\min
\left\{q_{\bm{\Gamma}^{(t+1)}}^{}(\bm{z}, \bm{z}^{\prime}), 
\frac{
\pi(\bm{z}^{\prime} \vert y) q_{\bm{\Gamma}^{(t+1)}}^{}(\bm{z}^{\prime}, \bm{z})
}
{
\pi(\bm{z} \vert y)
} \right\}
- \min
\left\{q_{\bm{\Gamma}^{(t)}}^{}(\bm{z}, \bm{z}^{\prime}), 
\frac{
\pi(\bm{z}^{\prime} \vert y)q_{\bm{\Gamma}^{(t)}}^{}(\bm{z}^{\prime}, \bm{z})
}
{
\pi(\bm{z} \vert y)
} \right\}. 
\end{split}
\end{equation}
It is easy to see that $\Delta^{(t)}$ is the difference of two Metropolis acceptance probabilities scaled by the proposal $q(\cdot)$ and the difference can be one of 4 possible values depending on which quantity is the minimum in each of the minimum operators. For ease of notation, relabel the inner terms of the minimum operator as $A,B,C,D$ so $\Delta^{(t)}$ as $\min\{A, B\} - \min\{C, D\}$. Considering these cases the 4 possible values for $\Delta^{(t)}$ are
\begin{equation}
\Delta^{(t)}
=
\begin{cases}
q_{\bm{\Gamma}^{(t+1)}}^{}(\bm{z}, \bm{z}^{\prime}) - q_{\bm{\Gamma}^{(t)}}^{}(\bm{z}, \bm{z}^{\prime}), & \text{if } A < B, C < D,\\[1.5ex]
\frac{\pi(\zvecpr \vert \datavec)}{\pi(\zvec \vert \datavec)}\left( q_{\bm{\Gamma}^{(t+1)}}^{}(\zvecpr, \zvec) - q_{\bm{\Gamma}^{(t)}}^{}(\zvecpr, \zvec) \right), & \text{if } A > B, C > D, \\[1.5ex]
\frac{\pi(\zvecpr \vert \datavec)}{\pi(\zvec \vert \datavec)}q_{\bm{\Gamma}^{(t+1)}}^{}(\zvecpr, \zvec) - q_{\bm{\Gamma}^{(t)}}^{}(\bm{z}, \bm{z}^{\prime}), & \text{if } A > B, C < D, \\[1.5ex]
q_{\bm{\Gamma}^{(t+1)}}^{}(\bm{z}, \bm{z}^{\prime}) - \frac{\pi(\zvecpr \vert \datavec)}{\pi(\zvec \vert \datavec)}q_{\bm{\Gamma}^{(t)}}^{}(\zvecpr, \zvec), & \text{if } A < B, C > D.
\end{cases}
\label{deltacases}
\end{equation}
For the first two cases of \eqref{deltacases} we only need to show that 
\begin{equation*}
\lvert q_{\bm{\Gamma}^{(t+1)}}^{}(\bm{z}, \bm{z}^{\prime}) - q_{\bm{\Gamma}^{(t)}}^{}(\bm{z}, \bm{z}^{\prime}) \rvert \text{ and } \lvert q_{\bm{\Gamma}^{(t+1)}}^{}(\zvecpr, \zvec) - q_{\bm{\Gamma}^{(t)}}^{}(\zvecpr, \zvec) \rvert \text{ both $\rightarrow$ 0 as } t \rightarrow \infty
\end{equation*}
which amounts to showing that $\lvert \bm{a}_i^{(t+1)} - \bm{a}_i^{(t)} \rvert \rightarrow 0$ and this will be shown below in equation.
For the third case of \eqref{deltacases} we can bound $\Delta^{(t)}$  from above as follows
\begin{equation*}
\frac{\pi(\zvecpr \vert \datavec)}{\pi(\zvec \vert \datavec)}q_{\bm{\Gamma}^{(t+1)}}^{}(\zvecpr, \zvec) - q_{\bm{\Gamma}^{(t)}}^{}(\bm{z}, \bm{z}^{\prime}) \leq q_{\bm{\Gamma}^{(t+1)}}^{}(\bm{z}, \bm{z}^{\prime}) - q_{\bm{\Gamma}^{(t)}}^{}(\bm{z}, \bm{z}^{\prime})
\end{equation*}
because the first term is bounded in the M-H Ratio (A > B) so this is equivalent to the first case of \eqref{deltacases}.
Finally for the final case of \eqref{deltacases}  we can bound the first term again using the M-H Ratio and factor out the likelihood factor as follows
\begin{equation*}
q_{\bm{\Gamma}^{(t+1)}}^{}(\bm{z}, \bm{z}^{\prime}) - \frac{\pi(\zvecpr \vert \datavec)}{\pi(\zvec \vert \datavec)}q_{\bm{\Gamma}^{(t)}}^{}(\zvecpr, \zvec) \leq \frac{\pi(\zvecpr \vert \datavec)}{\pi(\zvec \vert \datavec)} \left( q_{\bm{\Gamma}^{(t+!)}}^{}(\zvecpr, \zvec) - q_{\bm{\Gamma}^{(t)}}^{}(\zvecpr, \zvec) \right).
\end{equation*}

For all cases in \eqref{deltacases}, bounding $\lvert q_{\bm{\Gamma}^{(t+1)}}^{}(\bm{z}, \bm{z}^{\prime}) - q_{\bm{\Gamma}^{(t)}}^{}(\bm{z}, \bm{z}^{\prime}) \rvert$ is enough to establish diminishing adaptation. It therefore only necessary to prove that $\lvert \bm{a}_i^{(t+1)} - \bm{a}_i^{(t)} \rvert \rightarrow 0$ (or equivalently $\vert \bm{d}_i^{(t+1)} - \bm{d}_i^{(t)} \rvert \rightarrow 0$),  as $t \rightarrow \infty$.

Without loss of generality consider only the $\bm{a}^{(t)}_i$ parameter. The general form of the update scheme for $\bm{a}^{(t)}_i$ is
\begin{equation}
\log(\bm{a}_i^{(t+1)}) = \log(\bm{a}_i^{(t)}) + \frac{h}{t/n}(\alpha_\text{add} - \alpha_\text{target})
\end{equation}
and as $t \rightarrow \infty$, with $h<\infty$ and $\alpha_\text{target} < 1$
\begin{equation}
\left\lvert \log(\bm{a}_i^{(t+1)}) - \log(\bm{a}_i^{(t)}) \right\rvert \rightarrow 0.
\label{logdiff}
\end{equation}
To prove \eqref{logdiff} implies $\lvert \bm{a}_i^{(t+1)} - \bm{a}_i^{(t)} \rvert \rightarrow 0$, we must prove that the $\log$ function is bi-Lipschitz.
\noindent
Since $\log(x)$ and its inverse, $\exp{(x)}$ have a bounded first derivative, provided $0 < x < \infty$ which will be satisfied by the existence of $\epsilon > 0$, $\log$ is bi-Lipschitz (Definition \eqref{eq:bilip}).
Therefore
\begin{equation}
\lvert \bm{a}_i^{(t+1)} - \bm{a}_i^{(t)} \rvert \leq K \left\lvert \log(\bm{a}_i^{(t+1)}) - \log(\bm{a}_i^{(t)}) \right \rvert \rightarrow 0,
\end{equation}
and so diminishing adaptation for $\mathrm{P}_{\bm{\Gamma}^{(t)}}$ is established.
\end{proof}

\clearpage
\section{Results}
\label{sec:results}
We will now demonstrate our adaptive algorithm on a number of datasets, varying in size from a small to a large number of observations.
\begin{enumerate}
\item \textbf{Well Log Drilling data} - a small toy dataset to demonstrate the equivalence of filtering recursions and the adaptive changepoint sampler.
\item \textbf{Channel Noise data} - a moderately sized simulated dataset that takes minutes of precomputation for the filtering recursions, but seconds for our algorithm.
\item \textbf{Genome Variation data} - a large data set with over \SI{260000} observations, where it is not possible to use filtering recursions due to the presence of numerical error.
\end{enumerate}

For each of the datasets above, we compare our adaptive MCMC algorithm to the filtering recursions approach of \citet{raey}. We first take a long run of the filtering recursions at full precision and 
the posterior distribution from this long run is taken as the ground truth. Our adaptive MCMC algorithm is then compared to this ground truth by examining an approximate version of the Kullback-Leibler 
divergence in the posterior distribution of the number of changepoints over time. We define this measure of divergence as follows. Let $Q$ be the posterior distribution for the number of changepoints 
based on the ground truth (filtering recursions) and $P$ be the posterior distribution for the number of changepoints based on our adaptive MCMC algorithm. The divergence is defined as
\begin{equation}
\text{D}^{}_{\delta}(P \vert Q) = \sum_{k=0}^{n-1} \left\lbrack (1-\delta)P(i) + \delta\frac{1}{n} \right\rbrack \log \frac{(1-\delta)P(i) + \delta\frac{1}{n}}{(1-\delta)Q(i)+\delta\frac{1}{n}}.
\end{equation}
The correction parameter $\delta$ is necessary to ensure that the support of $P$ and $Q$ overlap and is chosen small ($ < 1\times10^{-10}$). Note that the Kullback-Leibler divergence results when $\delta=0$.

\subsection{Dataset 1 - Gaussian mean changepoint - Well Log Drilling data}
The problem of detecting changepoints in well log drilling data has been studied numerous times in the changepoint literature 
\citep{raey, ruanaidh2012numerical}. The well log drilling dataset originates from \citet{ruanaidh2012numerical} and consists of \SI{4050} probe 
measurements of the nuclear-magnetic response of underground rocks. The data was obtained by lowering the detection probe into a pilot drilled hole in 
the rock	 and recording the nuclear-magnetic response at discrete depth intervals. A changepoint is thought to occur when the rock type changes 
and such a change in signal is observed in the dataset. The data is shown in Figure \ref{welllog_data} with outliers removed as in \citet{raey}. These 
data have previously been analysed using filtering recursions to compute the posterior distribution of the number and position of changepoints. We will 
show that our algorithm reaches the same stationary distribution as the filtering recursions approach in the same time. The approximate filtering 
recursions using a lower level of precision will be also compared to our algorithm.
\begin{figure}[ht]
\centering
\includegraphics[width=\textwidth]{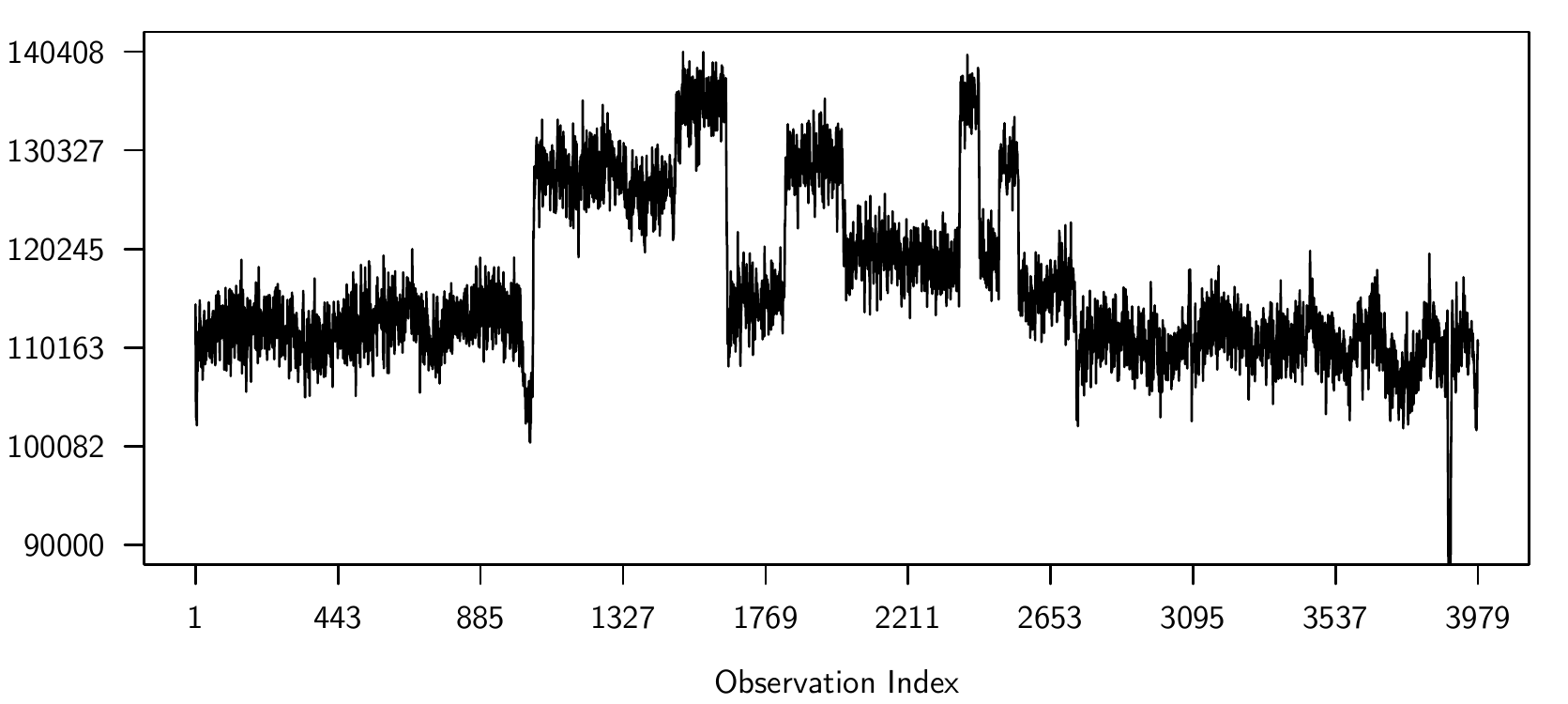}
\caption{\linespread{1.1}\selectfont{}\textbf{Well Log data -} The data consists of 3979 observations after outliers have been removed. Visually there are many changepoints in the data, prior to analysis.}
\label{welllog_data}
\end{figure}

\subsubsection{Well Log Drilling - Model}
We follow the approach of \citet{raey} by considering a Geometric ($p$ = 0.013) prior on the gap length between successive changepoints. The observations between changepoints are modelled as $\mathcal{N}(\mu_i, \sigma^2)$, where $\mu_i$ is the mean parameter for the $i$th segment and $\sigma$ is fixed to 2{,}500. Independent $\mathcal{N}(115{,}000, \tau^2 \sigma^2)$ priors are placed on each $\mu_i$ with $\tau^2$ set fixed to 16. Using the methods of Section \ref{sec:collapselike} the segment marginal likelihood can be shown (Appendix \ref{app:normallike}) to be
\begin{equation}
\mathrm{P}(a,b)=
(2\pi\sigma^2)^{-k/2}
(k\tau^2 + 1)^{-1/2} 
\exp 
\left( 
-\frac{1}{2\sigma^2} 
\left[
\left(s_2 - \frac{s_1^2}{k}\right) + \frac{k}{k\tau^2 + 1}\left(m - \frac{s_1}{k} \right)^2
\right]
\right).
\end{equation}
The quantities $s_1$ and $s_2$ are the sum and the sum of squares of the data $\{y_a, \dots y_b\}$, respectively.

\subsubsection{Well Log Drilling - Results \& Algorithm Comparison}
The results for the Well Log Drilling data run across adaptive MCMC, non-adaptive MCMC and filtering recursions are shown in Figure \ref{welllog_results}. The filtering recursions was run 
at 3 different levels of precision (full precision, 1e-6, 1e-4) to give 5 sets of results.
\begin{figure}[ht]
\centering
    \includegraphics[width=6.69423in,height=2.8in]{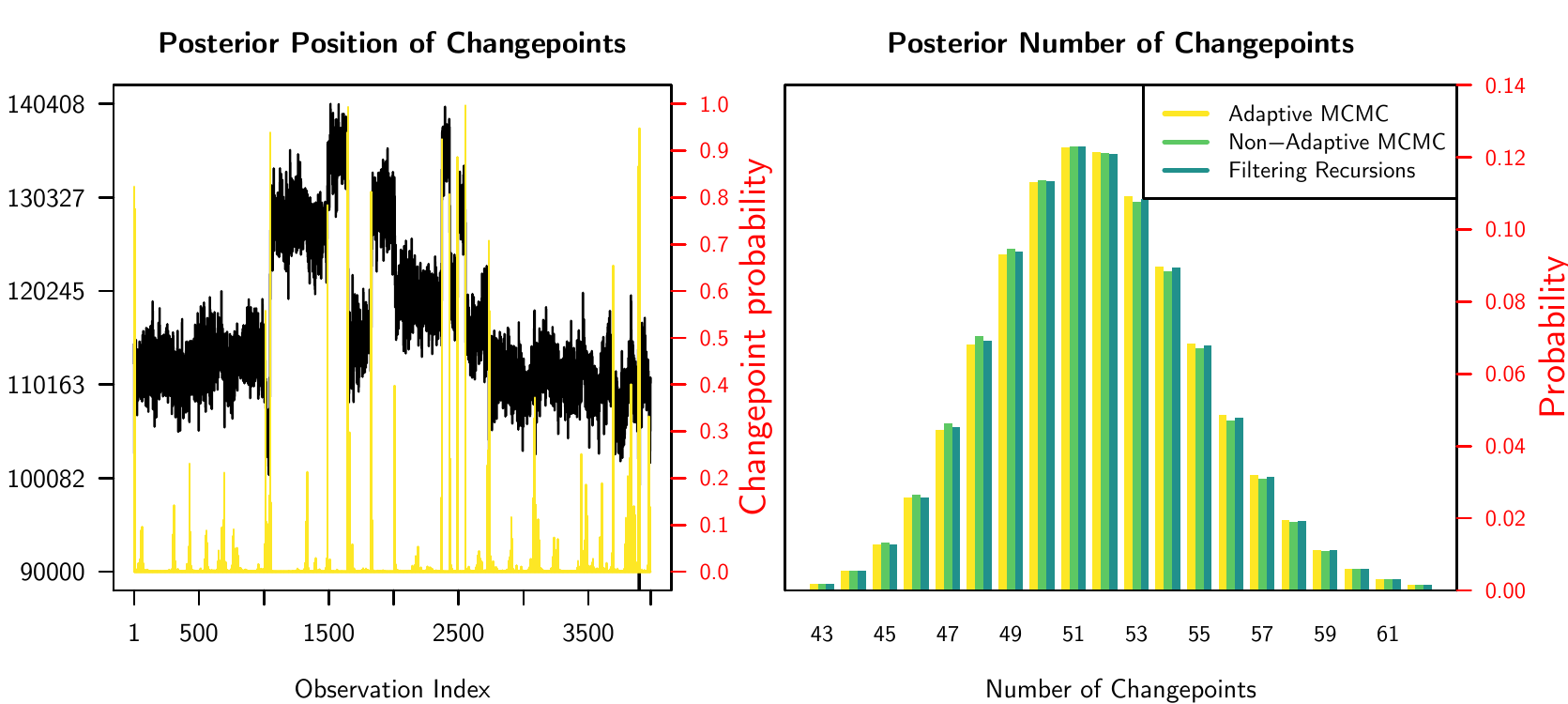}
\caption{\linespread{1.1}\selectfont{}\textbf{Well log data -} The left panel presents the estimated posterior probability of a changepoint at each observation based on the adaptive
MCMC algorithm. The right panel presents the estimated posterior probability of the number of changepoints for each of the adaptive MCMC, non-adaptive
MCMC and filtering recursions algorithms. This illustrates that each algorithm converges to the same stationary distribution.}
\label{welllog_results}
\end{figure}
The adaptive MCMC changepoint sampler was run for 5 seconds (\SI{16000000} iterations) with adaptive parameter $h=0.00119$ and a target acceptance 
rate of $15$\%. The non-adaptive MCMC sampler was run for 5 seconds (\SI{2000000} iterations).
The results in the right panel of Figure \ref{welllog_results} illustrate that the modal number of changepoints is estimated as 51 for all algorithms.
All algorithms capture the same posterior distribution of the number of changepoints and 
changepoint positions. 
Additionally, the left panel of Figure~\ref{welllog_results} displays the posterior position of changepoints from the adaptive MCMC run 
which (although not presented here) was very similar to the non-adaptive MCMC and filtering recursion algorithms.
The acceptance rates for the adaptive and non-adaptive MCMC algorithms were 15.31\% and 15.10\%, respectively and both the 
adaptive and non-adaptive MCMC algorithms were started from the same changepoint configuration, 40 changepoints randomly distributed throughout the data.

To compare the results of the adaptive MCMC changepoint sampler against filtering recursions, we compare the divergence of the adaptive and 
non-adaptive MCMC changepoint samplers to the output of filtering recursions run at full precision for 100 million chains from Carpenter's algorithm. 
For the 2 lower levels of precision (1e-6, 1e-4) in Figure \ref{welllog_div}, the filtering recursions algorithm fails to target the correct 
posterior once the precision level of the recursions equals $1\times 10^{-4}$. The adaptive MCMC changepoint sampler appears to converge marginally 
quicker to the target distribution, in the sense of reaching a low divergence, than the non-adaptive and filtering recursions algorithms. 
However we would overall recommend the use of filtering recursions for datasets of this size and smaller. The adaptive algorithm is marginally faster 
than the non-adaptive version and with a higher acceptance rate (15.31\%). This outlines that the Adaptive MCMC is competitive not only to the 
filtering recursions but also to the non-adaptive algorithm.
\begin{figure}[!ht]
\centering
\includegraphics[width=6.69423in]{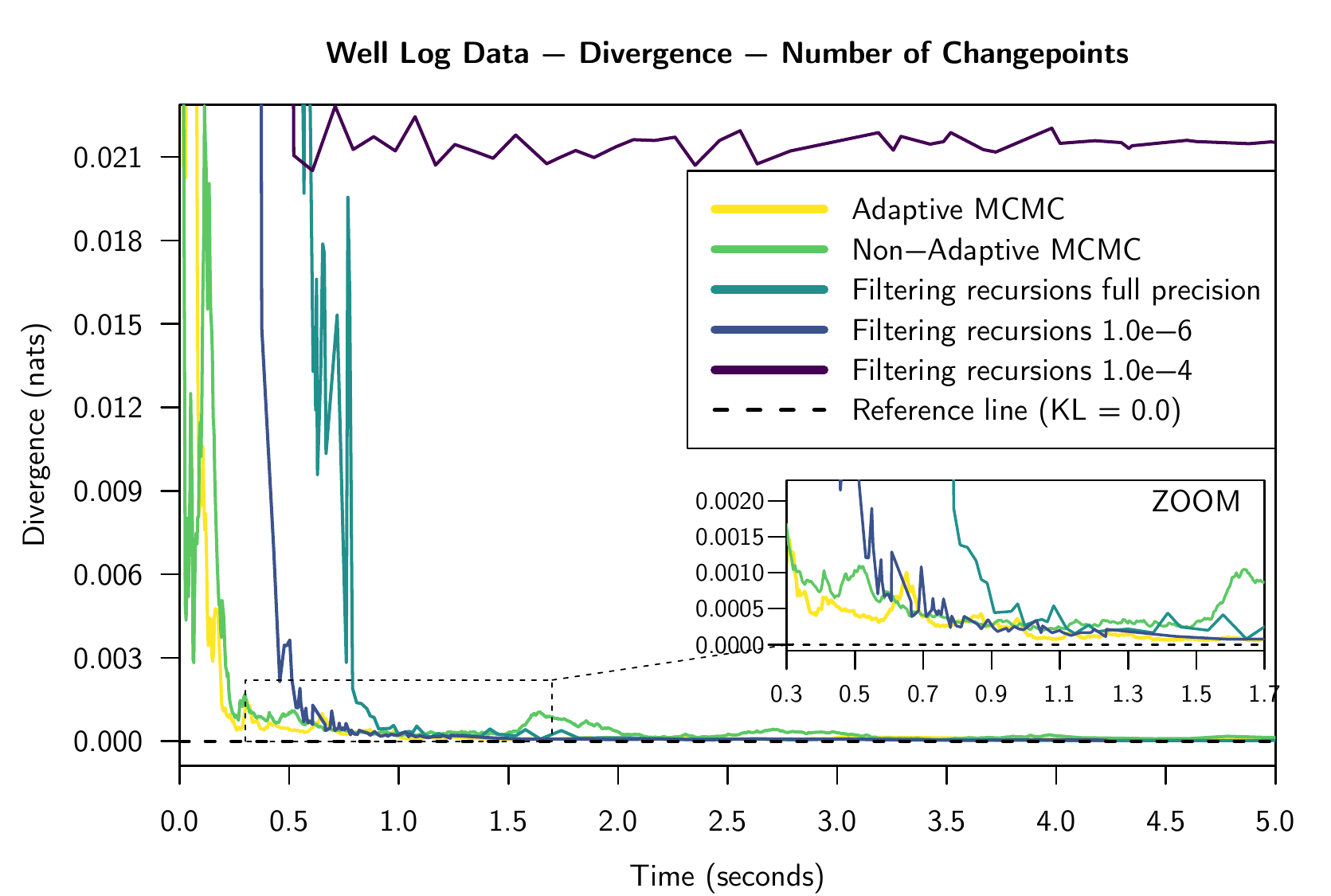}
\caption{\linespread{1.1}\selectfont{}\textbf{Well log data -} Divergence, $D_{\delta}$, between a precise estimate of the posterior distribution of the number of change points based on a long run of 
the filtering recursion algorithm to the adaptive and non-adaptive MCMC algorithm. This plot shows the convergence to the ground truth. All chains converge to the ground truth 
except for the low precision recursions. The adaptive algorithm is the most competitive of the MCMC algorithms.}
\label{welllog_div}
\end{figure}

\subsection{Dataset 2 - Gaussian precision changepoint - Channel Noise Data}
Variations in a signal can be detected by considering the change in variance around a fixed mean.
\begin{figure}[ht]
	\centering
	\includegraphics[width=\textwidth]{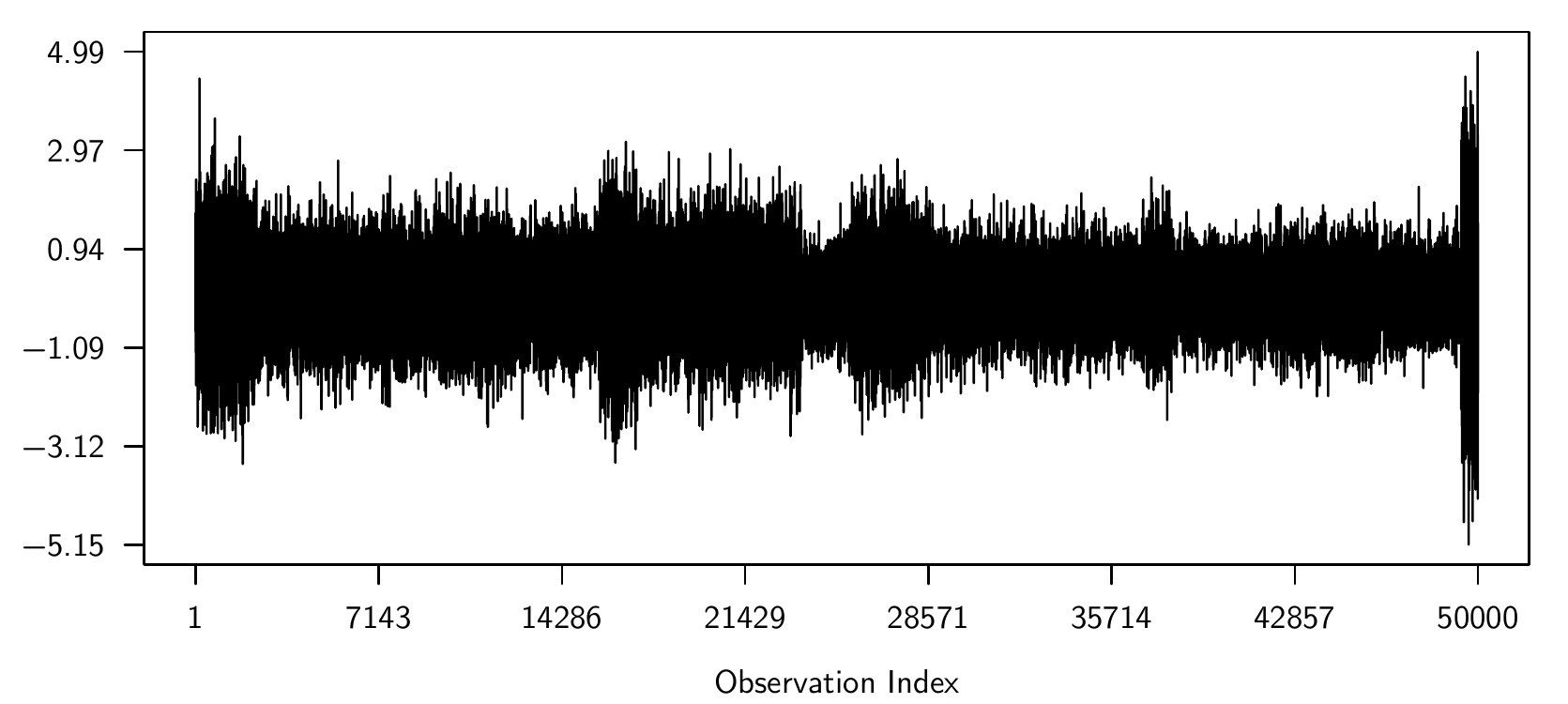}
	\caption{\linespread{1.1}\selectfont{}\textbf{Channel Noise data -} A simulated dataset where the variance is assumed to change over time around a fixed mean.}
	\label{channel_noise}
\end{figure}
 For example a web server may exhibit rapid variations in traffic across a period of time or a failing component of a machine may lose its precision 
 as it fails. If we assume a constant mean for each of the segments and allow there to be a change in precision $\lambda = \frac{1}{\sigma^2}$ at a 
 changepoint we can model a process such as shown in Figure \ref{channel_noise}.

The likelihood for each observation $y_i$ is $\mathcal{N}(\mu, \lambda^{-1})$ for a fixed $\mu$. Assuming a prior on the precision, 
$\lambda \sim \text{Gamma}(\alpha_0, \beta_0)$, allows one to integrate over $0 < \lambda < \infty$, leaving a marginal likelihood
\begin{equation}
\mathrm{P}(a,b) = \frac{(2\pi)^{-n/2}\Gamma(k/2 + \alpha_0)}{ \left(\beta_0 +\sum\limits_{i=a}^{b} (x_i - \mu)^2/2 \right)^{\alpha_0 + k/2}}, \hspace{2em} \text{where } k = b-a+1.
\label{variance_marg}
\end{equation}
For this data the hyperparameters were set to $\alpha_0 = 12.0, \beta_0 = 4.8$ and $\mu = 0$. The parameter $\mu$ can be set to 0 prior to 
analysis provided the data is shifted using its known mean. A geometric gap prior was placed on $\zvec$ with $p = 0.0006$.
\subsubsection{Results}
The results for the channel noise data are shown in Figure \ref{fig:var_change_results} for filtering recursions, the adaptive MCMC changepoint sampler 
and the non-adaptive MCMC changepoint sampler.
\begin{figure}[ht]
\centering
\includegraphics[width=6.69423in]{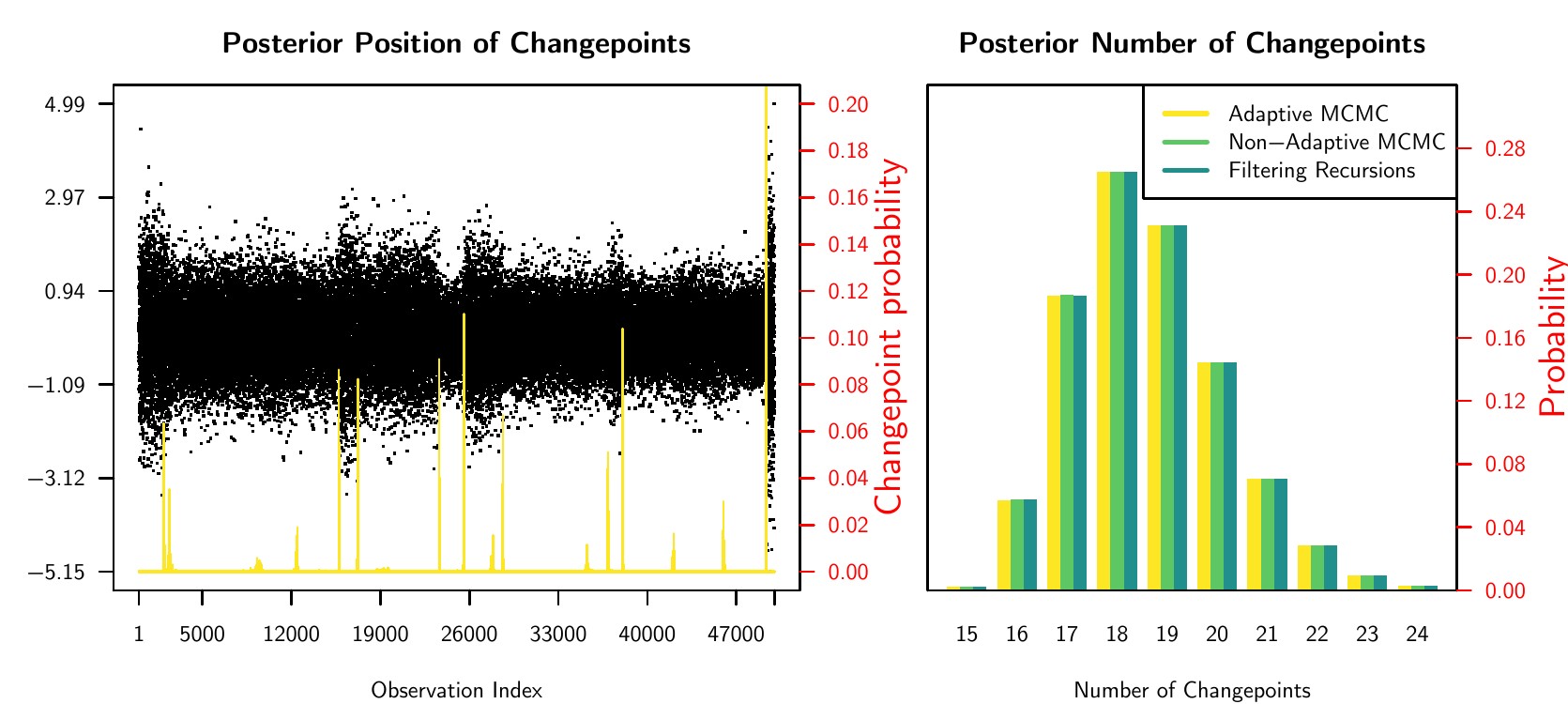}
\caption{\linespread{1.1}\selectfont{}\textbf{Channel Noise data -} The estimated posterior probability of a changepoint at each observation based on the adaptive
	MCMC algorithm is shown (left panel). The right panel presents the estimated posterior probability of the number of changepoints for each of the adaptive MCMC, non-adaptive
	MCMC and filtering recursions algorithms. This illustrates that each algorithm converges to the same stationary distribution.}
\label{fig:var_change_results}
\end{figure}
All 3 algorithms give a modal 18 number of changepoints in the data and each algorithm captures the full posterior distribution for both the positions and number of changepoints. The adaptive MCMC and non-adaptive MCMC algorithms were each run for 300 seconds, \SI{60000000} iterations and \SI{67000000} iterations respectively. The filtering recursions were run for the necessary precomputation 
time of $572$ seconds and then a further \SI{8500} seconds using Carpenter's algorithm (\SI{100000000} chains). Extra runs of the filtering recursions were run at precision levels (1e-12, 1e-10 and 1e-8) for comparison.
\begin{figure}[ht]
\centering
\includegraphics[width=6.69423in]{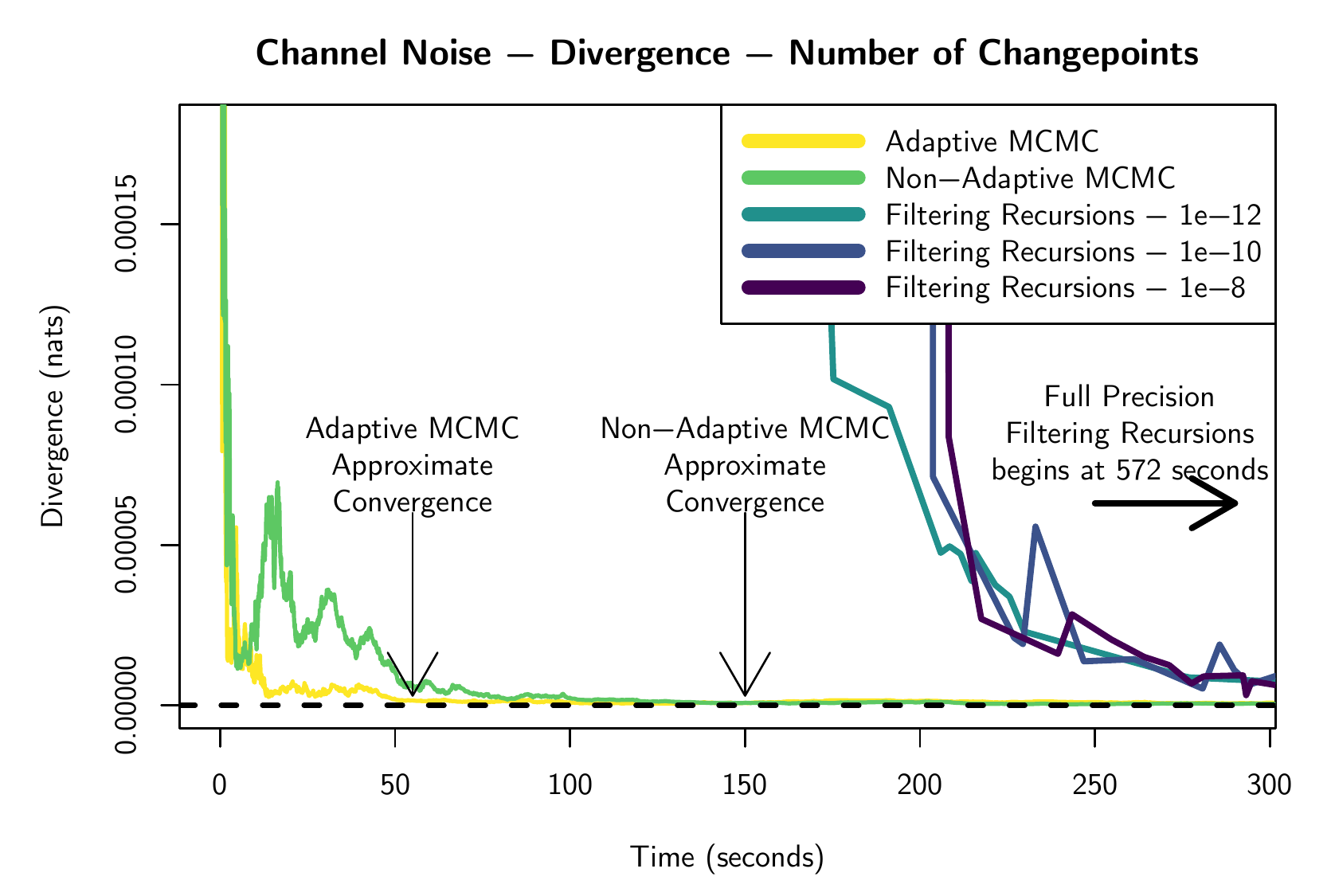}
\caption{\linespread{1.1}\selectfont{}\textbf{Channel Noise data -} Divergence, $D_{\delta}$, between a precise estimate of the posterior distribution of the number of change points based on a long run of 
the filtering recursion algorithm to the adaptive and non-adaptive MCMC algorithm. The adaptive MCMC algorithm outperforms non-adaptive MCMC and filtering recursions for all precision levels. We mark a 
convergence point for the adaptive MCMC of 53 seconds with divergence $1.43 \times 10^{-6}$ nats in the posterior distribution of the number of changepoints. The non-adaptive MCMC algorithm takes 150 seconds
to reach this level of divergence. Note that (although not shown on the plot) the full-precision filtering takes \SI{1738} seconds to reach this same level of divergence} 
\label{var_change_div}
\end{figure}

For this example, it is possible to compare the convergence properties of the adaptive MCMC, non-adaptive MCMC and filtering recursions. Due to the extended precomputation time required for the filtering recursions 
for datasets of this size, we can only compare the two algorithms after the precomputation has completed. We compare the algorithms by examining the time to converge to a certain level of Divergence, $D_{\delta}$. In 
Figure \ref{var_change_div} we mark an approximate convergence point at 53 seconds for the adaptive changepoint sampler with a divergence of $1.43 \times 10^{-6}$ nats. The filtering recursions is then 
run at full precision until it reaches this level of divergence or below, which takes \SI{1738} seconds. The results of this analysis are show in Figure \ref{var_change_div} and Table \ref{tab:divergence}, with Table \ref{tab:divergence} showing the relative speed of each algorithm. There is a vast in improvement using the adaptive MCMC algorithm.
\begin{table}[!ht]
	\centering
	\begin{tabular}{l|l|l|l|}
		\cline{2-4}
		& \textbf{Adaptive} & \textbf{Non-Adaptive} & \textbf{Filtering Recursions} \\ \hline
		\multicolumn{1}{|l|}{\textbf{Relative Speed}} & 1.0 (53 seconds)               & 2.83 (150 seconds)               &   32.8 (1738.64 seconds)                          \\ \hline
				\multicolumn{1}{|l|}{\textbf{Acceptance Rate}} & 12.1\%               & 10.9\%               &   -                          \\ \hline
	\end{tabular}
	\caption{\linespread{1.1}\selectfont{}\textbf{Channel Noise data -} Relative speed and acceptance rates of each (MCMC) algorithm are shown.}
	\label{tab:divergence}
\end{table}
% For this example this occurs after \SI{130000} iterations or $925$ seconds inclusive of precomputation time. This represents 
%an $91\%$ decrease in convergence time including precomputation. 
%For this example we recommend using our algorithm for faster convergence however the filtering recursions is still a viable alternative if fast convergence is not required. 

\subsection{Gaussian mean changepoint - Large Data example}
Genome variation profiling arises in the analysis of neuroblastoma (cancer) samples. The data comes in the form of DNA single nucleotide polymorphism (SNP) arrays. We analyse the $\log_2\text{ratioAB}$ series of sample GSM333824 UTP-N-12NMapping250KNsp from the SegAnnDB repository \citep{hocking2014seganndb}. The data consisting of \SI{262230} observations is displayed in Figure \ref{genome_plot}. Further details of the data and collection can be found in \citet{chen2008oncogenic}.
\begin{figure}[ht]
\centering
\includegraphics[width=\textwidth]{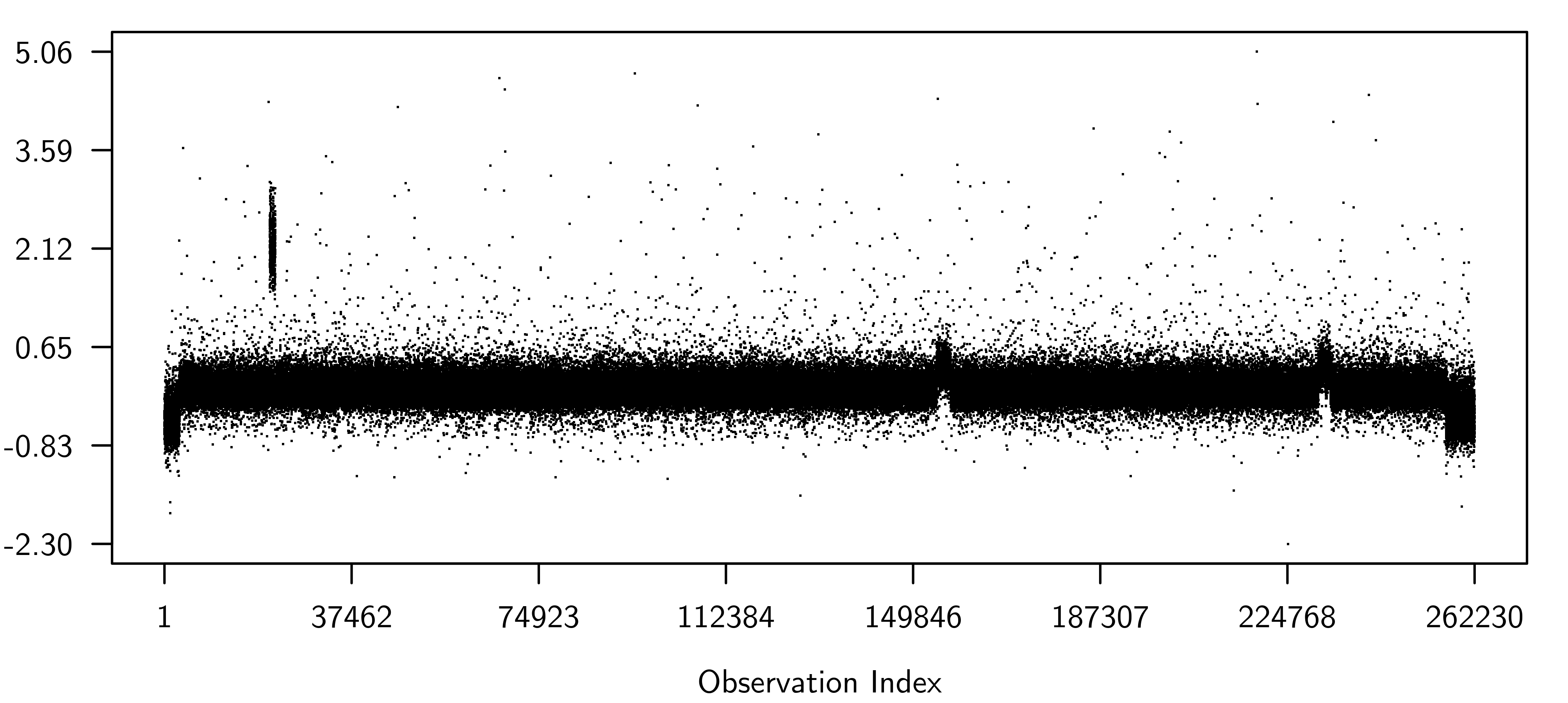}
\caption{\linespread{1.1}\selectfont{}\textbf{Genome variation dataset - }\SI{262230} observations from the neuroblastoma sample GSM333824 UTP-N-12NMapping250KNsp.}
\label{genome_plot}
\end{figure}

\noindent This is a large dataset and we analyse the change in the mean only. We assume a Geometric prior for the changepoint positions with parameter $p = 5.72 \times 10^{-5}$. The likelihood is take as $\mathcal{N}(\mu_j, 0.13)$ and the prior for $\mu_j$ is $\mathcal{N}(0, 116.0)$.

\citet{hocking2014seganndb} analysed these data sequences using the PrunedDP algorithm of \citet{rigaill2015pruned}. Their algorithm works by fixing the number of changepoints $1 < k < k_\text{max}$ in such a way to minimise the least squared error of all possible segmentations using $k+1$ segments. The value of $k_\text{max}$ is chosen as low as 20. We find for our analysis we find between $219$ and $260$ changepoints using the adaptive changepoint sampler. 
\subsubsection{Difficulty with filtering recursions}
\label{sec:diffrecur}
For a dataset of this size the numerical stability of the filtering recursions can cause problems. In the calculation of the transition probabilities \eqref{transprob}, which are needed to sample from the recursions using the Carpenter's algorithm \citep{carpenter1999improved}, there is potential to encounter numerical errors arising from building the forward proposal distribution of the next changepoint. For this dataset, considering every possible changepoint location $j \in \{1, \dots n-1\}$, there will be a maximum $\binom{262,230}{2} = 3.43 \times 10^{10}$ transition probabilities $\mathrm{P}(\tau_j \vert \tau_{j-1})$ to calculate, see equation \eqref{transprob}.
Many of the $\mathrm{P}(\tau_r \vert \tau_{j-1})$ terms will be very small with values less than subnormal machine precision even when computed on the $\log$ scale. Numerically, transition probabilities close to 0 are regarded as having negligible contribution to proposing changepoints and will not be sampled. However for datasets where changepoints are far apart i.e.\ $\tau_r \gg \tau_{j-1}$, the calculation of the cumulative distribution which is needed to propose the next changepoint,
\begin{equation}
\mathrm{P}(\tau \leq \tau_r \vert \tau_{j-1}) = \log \left(e^{\mathrm{P}(\tau_r \vert \tau_{j-1})} + e^{\sum_{i=1}^{r-1} \mathrm{P}(\tau_i \vert \tau_{j-1})} \right),
\label{logsum}
\end{equation}
will not correctly accumulate all of these small probabilities. Since the number of small probabilities is significantly large, this leads to those probabilities greater than subnormal machine probabilities to be artificially inflated relative to the magnitude they would normally appear had the small probabilities being accumulated to infinite precision. The effect of this is that these points will be chosen as changepoints more frequently than they should be and points with small probabilities never to be chosen even though taken together they consume non-negligible mass of the transition distribution. This agrees empirically with our analysis for this dataset and other even larger datasets.
\subsubsection{Results for the Adaptive Algorithm}
The adaptive algorithm was run for \SI{40000000000} iterations with \SI{20000000000} iterations removed by burn-in. The long length of burn-in was necessary by the analysis in Figure \ref{fig:mapovertime}.
The adaptive parameter $h$ was set to 0.001, while a target acceptance rate of 1.0\% was chosen to tune the adaptive scheme. The algorithm took \SI{20665} seconds on an Intel i7 3.40GHz and the achieved acceptance rate was 1.048\%. Over 210 changepoints are detected by the adaptive changepoint sampler. This includes the changepoints detected by the SegAnnDB repository and more changepoints at other locations. The PrunedDP algorithm of \citet{rigaill2015pruned} 
requires the user to choose an optimal $k$ in a post processing step. This contrasts with our adaptive algorithm which presents the user with the posterior distribution of the number of changepoints $k$, 
and so allowing the user to examine the \textit{a posteriori} probability of various $k$ and examine the relevant strength of different segmentations.
\begin{figure}[ht]
	\centering
	\includegraphics[width=6.69423in, height=3in]{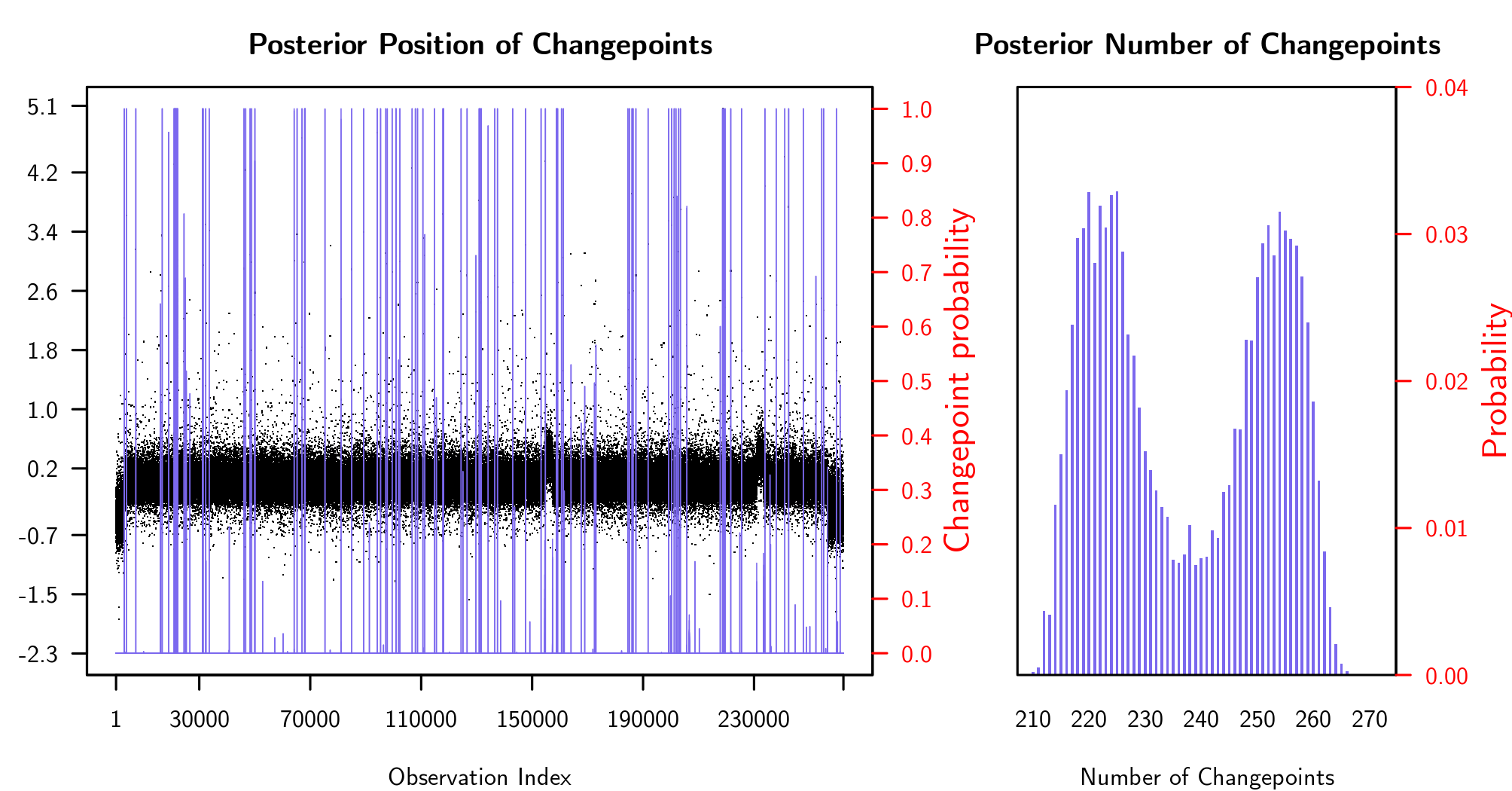}
	\caption{\linespread{1.1}\selectfont{}\textbf{Genome variation dataset -} The adaptive algorithm captures a bimodal posterior distribution for the number of changepoints (right panel). The posterior distribution of the changepoints positions is presented in the left panel.}
	\label{fig:largedata_posterior}
\end{figure}
\begin{figure}[!ht]
	\centering
	\includegraphics[width=6.69423in]{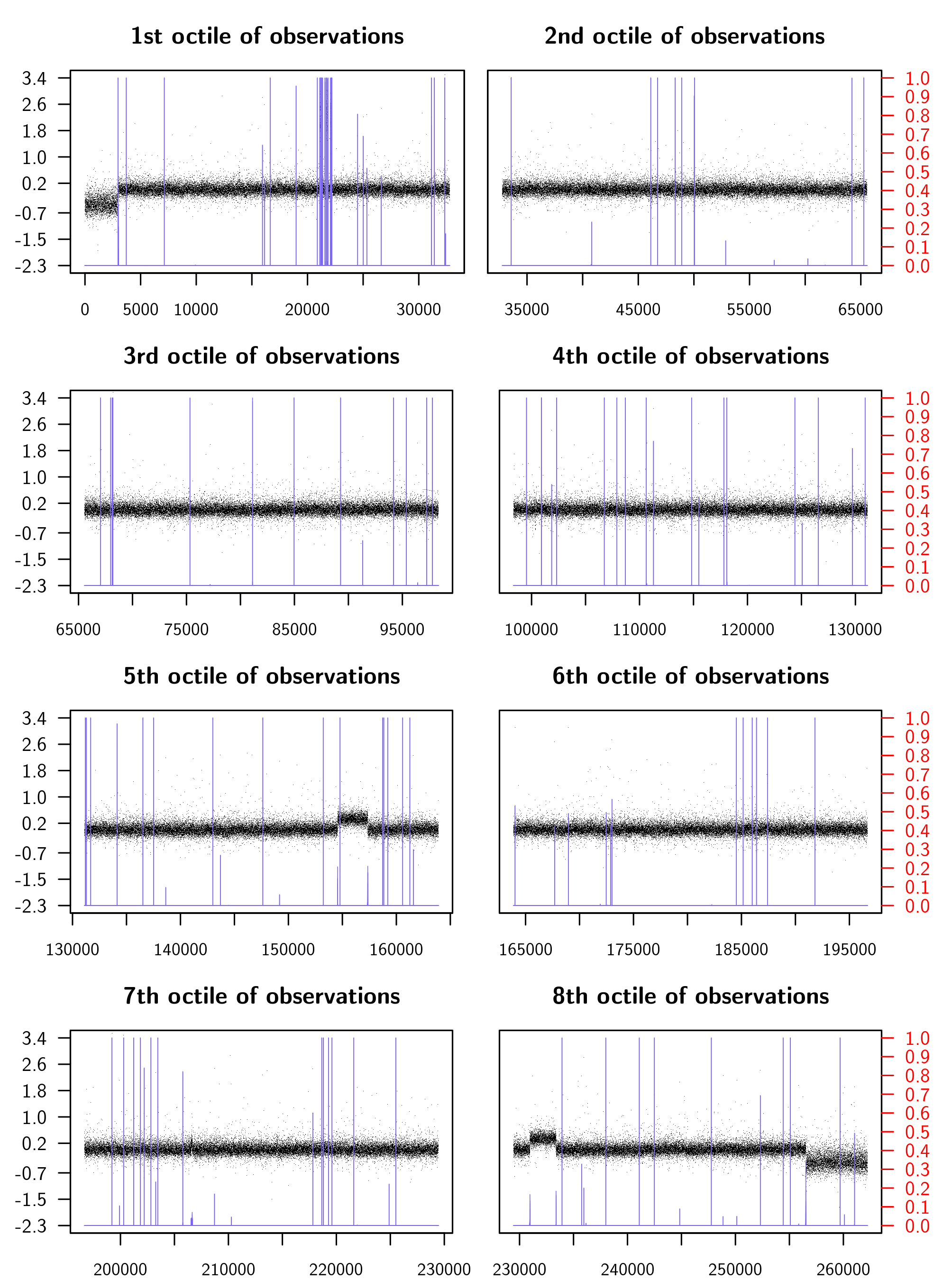}
	\caption{\linespread{1.1}\selectfont{}\textbf{Genome variation dataset -} More detailed examination of the location of changepoints. In the first octile, high changepoint activity is observed. In the first half of the seventh octile, many outliers are detected by the adaptive algorithm.}
	\label{fig:detailed_largedata}
	\end{figure}
\clearpage
\subsubsection{Algorithm Comparison - Adaptive and Non-Adaptive MCMC}
It is not possible to run the filtering recursions for this data due to issues discussed in Section \ref{sec:diffrecur}. In particular, and in contrast with the two previous examples, it is not possible to assess how well
each algorithm converges to the target posterior distribution. However we can still provide a good indication of the convergence of each of the adaptive and non-adaptive MCMC algorithms by exploring the trajectory of 
the state of each chain towards the maximum \emph{a posteriori} of the target distribution. We performed a 6 hour run of both algorithms for over \SI{80000000000} iterations and monitored the maximum a posteriori (MAP) 
estimate achieved. The results are shown in Figure \ref{fig:mapovertime}.
\begin{figure}[ht]
	\includegraphics[width=\textwidth]{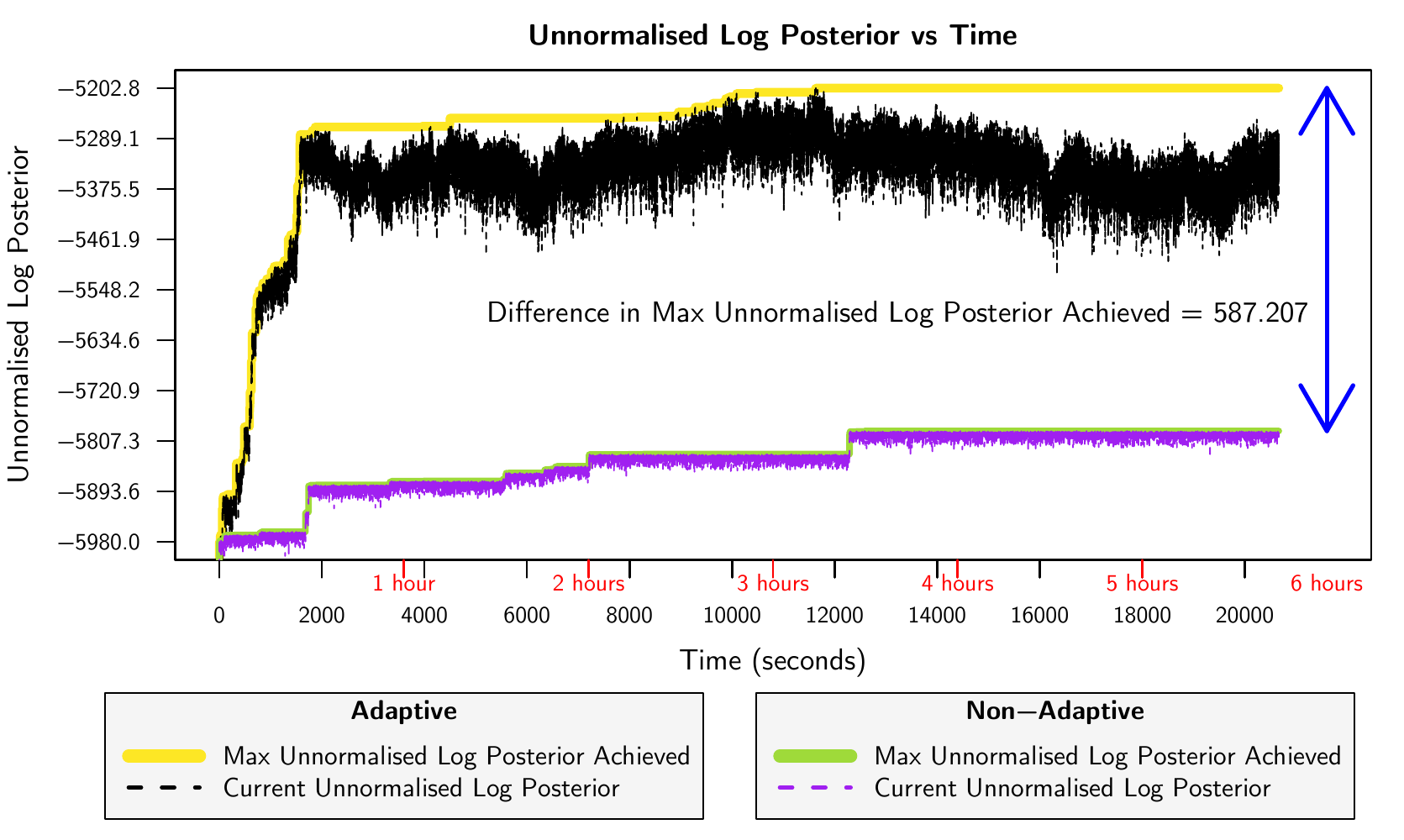}
%\caption{Comparison between the adaptive and non-adaptive with reference to a 3 hour run ($50\mathrm{e}{9}$ iterations) of the non-adaptive algorithm. The adaptive algorithm shows a large improvement over the non-adaptive version. }
	\caption{\linespread{1.1}\selectfont{}\textbf{Genome variation dataset -} Comparison between the trajectory of the log unnnormalised posterior for the adaptive and non-adaptive algorithms. 
	It is clear that the adaptive algorithm climbs to an area of high posterior probability many times faster than the non-adaptive algorithm.}
	\label{fig:mapovertime}
\end{figure}
The Adaptive algorithm reaches a higher area of the posterior after about \SI{2000} seconds and continues to find higher areas of posterior mass. In constrast the non-adaptive version is much slower to climb to this 
area and after 6 hours still had not reached the MAP estimate of the Adaptive algorithm. This gives some indication that the adaptive MCMC algorithm is better able to reach the high-posterior density regions than
the non-adaptive MCMC algorithm. We therefore conclude for datasets of this size, that the adaptive algorithm is many times more competitive than the non-adaptive algorithm and due to filtering recursion being 
unavailable is an ideal algorithm for big data changepoint problems.

\section{Conclusion \& Discussion}
This paper introduces an adaptive changepoint sampling algorithm for multiple changepoint problems. 
%which are collapsible, whereby all model parameters can be integrated out, apart from the indicator variables which denote points at which changepoint occur. 
We have described how our algorithm is be designed to learn \textit{on-the-fly} where changepoints 
are likely to be located in a dataset. We prove that the adaptive MCMC scheme which we develop leaves the posterior distribution ergodic. Moreover the adaptive 
MCMC algorithm scales to large datasets in contrast to the filtering recursions of \citet{raey} which is unreliable and prone to numerical instability in 
this case. Three datasets increasing in size from \SI{4000} observations to over \SI{260000} observations have been illustrated in this paper. The latter and 
largest dataset is unable to be analysed using filtering recursions and we show that our algorithm works well here to detect the number and location of 
changepoints in a reasonable computational time. We recommend using the filtering recursions for smaller datasets (e.g.\ up to size \SI{100000} observations) and where computational time is not an issue. 
However for datasets with more than \SI{100000} observations we advocate using the adaptive 
changepoint sampler. 

Further work will involve extending this adaptive MCMC approach to other posterior distributions on discrete 
state spaces. For example, the likelihood of the data in this paper assumes independent observations within a segment between two changepoints. This could be replaced with a dependence within segment 
likelihood as in the work of \citet{wyse2011approximate} where the marginal segment likelihood is replaced with integrated nested Laplace approximations.

%This paper has illustrated that adaptive MCMC can be usefully applied to discrete state spaces problems. Problems on discrete spaces abound in statistics such as Bayesian optimisation, variable selection and 
%latent variable models. The collapsed stochastic block model \citep{mcdaid2013improved} and other models where conjugate priors allow exact expressions for the marginal likelihood of the data with respect 
%to the model parameters and to leave a discrete state space posterior.

The diminishing adaptation condition we have proved in this paper is just one method 
of automatically tuning adaptive proposals. Our adaptation condition takes the form of a stochastic approximation algorithm but more involved adaptation schemes may be designed using the theory developed in this paper
and this is a focus of future work. To conclude, we feel that there is much wider scope for the implementation of adaptive MCMC in practice and we hope that this article will encourage more work in this direction. 

%An important part of this work is that the advantage of adaptive MCMC for problems that are computationally challenging has been illustrated. For Bayesian methods to work well in high-dimensional problems it is 
%necessary that the methods scale well with the size of the data.

\section*{Acknowledgements}\label{sec:Acknowledgements}
The Insight Centre for Data Analytics is supported by Science Foundation Ireland under Grant Number SFI/12/RC/2289.
Alan Benson and Nial Friel's research was also supported by a Science Foundation Ireland grant: 12/IP/1424.

\clearpage
\bibliography{adaptive_changepoint_alanbenson_nialfriel.bib}
\clearpage
\appendix
\titlelabel{\thetitle.\quad}
\section*{Appendices}
\section{Normal Marginal Likelihood Calculation}
\label{app:normallike}
The marginal likelihood for a changepoint in the mean parameter for normally distributed data with known variance ($\sigma^2$) and with a $\mathcal{N}(\mu, \tau^2\sigma^2)$ prior on $\mu$ can be expressed with $k = b-a+1$ as the integral of the product of two normal densities
\begin{equation}
\mathrm{P}(a,b) = \int_0^{\infty}  
\frac{(2\pi \sigma^2)^{-(k+1)/2}}{\tau} 
\prod_{i=a}^{b} 
\exp\left( 
-\frac{1}{2\sigma^2} \left[
\left(k+\frac{1}{\tau^2}\right)\mu^2 -2 \left( s_1 + \frac{m}{\tau^2} \right)\mu + \left(s_2 + \frac{\mu^2}{\tau^2}\right)
\right]
\right)
\, d \mu
\end{equation}
where $s_1 = \sum_{i=a}^{b} y_i$ and $s_2 = \sum_{i=a}^{b} y_i^2$. Completing the square and rearranging
\begin{equation}
= 
(2\pi\sigma^2)^{-k/2}
(k\tau^2 + 1)^{-1/2} 
\exp 
\left( 
-\frac{1}{2\sigma^2} 
\left[
\left( 
s_2 + \frac{\mu^2}{\tau^2} 
\right)
- \frac{\tau^2}{k\tau^2 + 1}\left( s_1 + \frac{\mu}{\tau^2} \right)^2
\right]
\right) 
\label{unsafe}
\end{equation}
completing the square again with the term inside the square brackets gives
\begin{equation}
 = 
(2\pi\sigma^2)^{-k/2}
(k\tau^2 + 1)^{-1/2} 
\exp 
\left( 
-\frac{1}{2\sigma^2} 
\left[
\left(s_2 - \frac{s_1^2}{k}\right) + \frac{k}{k\tau^2 + 1}\left(m - \frac{s_1}{k} \right)^2
\right]
\right)
\end{equation}
This is a more numerically stable version than \eqref{unsafe} as $s_1 - \frac{s_1^2}{k}$ is the sum of squared deviations from the segment sample mean which can be calculated recursively and $m - \frac{s1}{k}$ is the distance of the segment sample mean from the prior which will cause no numerical issues.

\section{Walker's Alias Method with Vose's Correction}
\label{app:alias}
The Alias Method is due to \citet{walker1974new} and the numerical safe approach to constructing Alias tables, which are needed for this method, is due to \citet{vose1999simple}. The algorithm is a very simple approach to simulating form a general categorical distribution with $k$ categories each having a (possibly unnormalised) weight $w_k$.

The weights are first normalised and then two tables are constructed, a probability table and an Alias table. Some of the normalised weights will be greater than the average probability $\frac{1}{k}$ and are known as Big Points, and some will be less than or equal to it, the Small Points. The method works by moving some of the probability mass from the Big Points to the Small Points. All Small Points will eventually be associated with at most one of the Big Points (its alias).

Once the Alias table has been constructed they can be sampled from in $\mathcal{O}(1)$ time. Simply select a Small Point uniformly at random and then use a biased coin flip to choose either that point or its Alias point. This method is extremely efficient and is currently the best of all methods for sampling from finite categorical distributions however if $w_k$ changes for any $k$ the entire tables must be reconstructed in $\mathcal{O}(n)$ more steps. Another method with a similar computational efficiency is the Gumbel Max Method \citep{yellott1977}

\section{Carpenter's Algorithm}
\label{app:carpenter}
Carpenter's Algorithm \citep{carpenter1999improved} is a method of sampling from a discrete probability distribution similar to the Alias method but without the need for precomputed probability tables. It works by exploiting the fact that the spacing in the uniform distribution on [0,1) is exponential with rate $1$.
To sample $n$ values from $x = \{1, \dots M\}$ with $P(X = i) = p_i$
\begin{enumerate}
\item Simulate $e_1, \dots, e_{n+1} \sim \exp{\lambda = 1}$
\item Create the step function (CDF) $u_j = \frac{\sum_{i=1}^{j} e_i}{\sum_{i=1}^{n+1} e_i}$ for $j=1, \dots n+1$
\item Set $Q = 0$, $U = u_1$, $j=1$, $i=1$
\item If $U < Q + P(X = j)$ output $j$ and set $U = u_{i+1}$ and $i =i +1$. Otherwise set $Q = P(X = j)$ and $j = j+1$. Repeat until $i = n+1$.
\end{enumerate}

\section{Dual Adaptation}
\label{app:dual}
Only one of the adaptive parameters for a point $i$ ($\bm{a}_i$ / $\bm{d}_i$) are updated when either an add or delete move at this point has been accepted. \citet{griffin2014individual} has suggested that information can still be gained for both $\bm{a}_i$ and $\bm{d}_i$ regardless of which move has been performed.

Dual adaptation involves using the M-H ratio calculated for the current move, denoted $\alpha_F(\zvec, \zvec^{\prime})$ for the \textit{forward} move, and its \textit{reverse} move, denoted $\alpha_R(\zvec^{\prime}, \zvec)$. Calculation of $\alpha_R$ is trivial once $\alpha_F$ is available.
\citet{griffin2014individual} shows how to modify the adaptation scheme so that it continues to target $M$.
\noindent
The average \textit{a posteriori} mutation rate of the algorithm is 
\begin{equation*}
M = \int C(\zvec, \zvec^{\prime})\alpha(\zvec, \zvec^{\prime}) q(\zvec, \zvec^{\prime})\pi(\zvec \vert y) \, d \zvec \, d \zvec^{\prime} \hspace{2em} 
\end{equation*}
where  $q(\zvec, \zvec^{\prime})$ depends on the move (add / delete)
 and  $C(\zvec, \zvec^{\prime}) = 0 $  if  $z_i = z^{\prime}_i \hspace{0.5em} \forall i$.

If we wish to continue targeting this mutation rate under Dual adaptation we need to define a second chain to preserve detailed balance.
\begin{equation*}
(\bm{\delta}, \bm{\delta}^\prime) = 
\begin{cases}
(\zvec^{\prime}, \zvec), & \text{with probability } \alpha (\zvec, \zvec^{\prime}), \\[1.5ex]
(\zvec, \zvec^{\prime}), &\text{with probability } 1 - \alpha (\zvec, \zvec^{\prime}).
\end{cases}
\end{equation*}
Now
\begin{equation*}
\begin{split}
M_{\deltavec} =&\int C(\deltavec, \deltavecpr)\alpha(\deltavec, \deltavecpr) q(\deltavec, \deltavecpr)\pi(\deltavec \vert y) \, d \deltavec \, d \deltavecpr \\[1.5ex]
=& \hspace{0.25em} \mathbf{E} \lbrack C(\deltavec, \deltavecpr)\alpha(\deltavec, \deltavecpr) \rbrack \\[1.5ex]
=& \hspace{0.25em} \alpha (\zvec, \zvec^{\prime})\mathbf{E}(C(\zvec^{\prime}, \zvec)\alpha(\zvec^{\prime}, \zvec)) + (1-\alpha (\zvec, \zvec^{\prime}))\mathbf{E}(C(\zvec, \zvec^{\prime})\alpha(\zvec, \zvec^{\prime}))
\end{split}
\end{equation*}
and weighting this with the original mutation rate we get
\begin{equation*}
w\alpha (\zvec, \zvec^{\prime})\mathbf{E}(C(\zvec^{\prime}, \zvec)\alpha(\zvec^{\prime}, \zvec)) + (1-w\alpha (\zvec, \zvec^{\prime}))\mathbf{E}(C(\zvec, \zvec^{\prime})\alpha(\zvec, \zvec^{\prime}))
\end{equation*}
note $ C(\zvec^{\prime}, \zvec) = C(\zvec, \zvec^{\prime})$

The new adaptive scheme becomes

\bigskip
If an add move has just been accepted:
\begin{align*}
&\text{Update $\bm{a}_i$}& \hspace{1em} \log(\bm{a}_{i}^{(t+1)})
&= \log(\bm{a}_{i}^{(t)}) + \left(\frac{h}{t/n}\right)
\left( \alpha(\zvec, \zvec^{\prime}) - \alpha_\text{target} \right)\left( 1- w\alpha(\zvec, \zvec^{\prime}) \right). \\[1.5ex]
&\text{Update $\bm{d}_i$}& \hspace{1em} 
\log(\bm{d}_{i}^{(t+1)})
&= \log(\bm{d}_{i}^{(t)}) + \left(\frac{h}{t/n}\right)
\left( \alpha(\zvec^{\prime}, \zvec) - \alpha_\text{target} \right)\alpha(\zvec, \zvec^{\prime}).
\end{align*}

\bigskip
If a delete move has just been accepted:
\begin{align*}
&\text{Update $\bm{a}_i$}& \hspace{1em} \log(\bm{a}_{i}^{(t+1)})
&= \log(\bm{a}_{i}^{(t)}) + \left(\frac{h}{t/n}\right)
\left( \alpha(\zvec^{\prime}, \zvec) - \alpha_\text{target} \right)\alpha(\zvec, \zvec^{\prime}). \\[1.5ex]
&\text{Update $\bm{d}_i$}& \hspace{1em} \log(\bm{d}_{i}^{(t+1)})
&= \log(\bm{d}_{i}^{(t)}) + \left(\frac{h}{t/n}\right)
\left( \alpha(\zvec, \zvec^{\prime}) - \alpha_\text{target} \right)\left( 1- w\alpha(\zvec, \zvec^{\prime}) \right).
\end{align*}
The choice of $w$ is recommended as 0.5 by \citet{griffin2014individual}.

\end{document}